\documentclass[manyauthors]{fundam}
\usepackage{url} 
\usepackage[ruled,lined]{algorithm2e}
\usepackage{graphicx}

\usepackage{xcolor}
\usepackage{colortbl}
\usepackage{stmaryrd,amssymb,mathtools}
\usepackage{microtype}
\usepackage{verbatim}
\usepackage{paralist}
\usepackage{multirow}
\usepackage[colorinlistoftodos, textwidth=3cm]{todonotes}
\usepackage{latexsym}
\usepackage{amsmath}
\usepackage{paralist}
\usepackage{verbatim}
\usepackage{mathpartir}
\usepackage{soul}
\usepackage{amsfonts,amssymb,mathtools}
\usepackage[all]{xy}
\usepackage[normalem]{ulem}

\usepackage{orcidlink}

\hypersetup{
	colorlinks,
	citecolor=black,
	filecolor=black,
	linkcolor=black,
	urlcolor=black
}

\newcommand{\Env}{\mathit{E}}

\newcommand{\tletr}{{\tt letrec}}

\newcommand{\tin}{{\tt in}}

\newcommand{\wrt}{{w.r.t.\ }}
\newcommand{\eg}{{e.g.\ }}

\newcommand{\ari}{\mathit{ar}}
\newcommand{\env}{\mathit{env}}
\newcommand{\gentzen}[2]{{\displaystyle \frac{#1}{#2}}}
\newcommand{\gentzent}[3]{{\displaystyle \frac{#1}{#2}}\quad #3}
\newcommand{\dotcup}{\ensuremath{\mathaccent\cdot\cup}}
\newcommand{\matcheq}{{\ \unlhd\ }}

\newcommand{\MBA}{\mathbb{A}}
\newcommand{\LRL}{\mathit{LRL}}
\newcommand{\LRLX}{\mathit{LRLX}}

\newcommand{\LRLXAE}{\mathit{LRLXAE}}

\newcommand{\FA}{\mathit{FA}}
\newcommand{\BA}{\mathit{BA}}
\newcommand{\AT}{\mathit{AT}}
\newcommand{\LA}{\mathit{LA}}
\newcommand{\tops}{\mathit{tops}}

\newcommand{\FIX}{\mathit{FIX}}
\newcommand{\Var}{\mathit{Var}}
\newcommand{\AtVar}{\mathit{AtVar}}

\newcommand{\dom}{\mathit{dom}}
\newcommand{\letrecunify}{{\sc {LetrecUnify}}}
\newcommand{\letrecmatch}{{\sc {LetrecMatch}}}
\newcommand{\letrecenvmatch}{{\sc {LetrecEnvMatch}}}

\newcommand{\letrecdagmatch}{{\sc {LetrecDagMatch}}}
\newcommand{\freshdot}{{\#}}

\newcommand{\NP}{\mathit{NP}}
\newcommand{\gen}[1]{\langle #1 \rangle}

\newcommand{\numberSize}{{\#\mathrm{Lr}\lambda\mathrm{FA}}}
\newcommand{\numberVar}{{\#\mathrm{Var}}}

\newcommand{\numberEqsNonX}{{\#\mathrm{EqNonX}}}  
\newcommand{\numberEquations}{{\#\mathrm{Eqs}}}

\newcommand{\LRLXA}{\mathit{LRLXA}}
\newcommand{\letrecunifyA}{{\sc {LetrecUnifyAV}}}
\newcommand{\letrecunifyAB}{{\sc {LetrecUnifyAVB}}}

\newcommand{\val}{\mathit{val}}
\newcommand{\inv}{\mathit{inv}}
\newcommand{\ELIMAB}{\text{ElimAB}}
\newcommand{\bbbr}{\mathbb{R}}

\definecolor{upmaroon}{rgb}{0.48, 0.07, 0.07}

\newcommand{\ignore}[1]{}

\usepackage{hyperref}

\begin{document}

\setcounter{page}{247}
\publyear{22}
\papernumber{2110}
\volume{185}
\issue{3}

  \finalVersionForARXIV

\title{Nominal Unification and Matching of \\ Higher Order Expressions with Recursive Let}


\author{Manfred Schmidt-Schau{\ss}\orcidlink{https://orcid.org/0000-0001-8809-7385}\thanks{Address
                         for correspondence: Goethe-Universit\"at Frankfurt am Main
                               Fachbereich 12: Informatik und Mathematik, Robert-Mayer-Stra{\ss}e 11-15, 60325
                               Frankfurt am Main, Grrmany. \newline \newline
          \vspace*{-6mm}{\scriptsize{Received February 2021; \ accepted April 2022.}}}
 \\
GU Frankfurt, Germany\\
schauss@ki.cs.uni-frankfurt.de
\and Temur Kutsia\orcidlink{https://orcid.org/0000-0003-4084-7380}\\
RISC, JKU Linz, Austria\\
kutsia{@}risc.jku.at
\and
Jordi Levy\orcidlink{https://orcid.org/0000-0001-5883-5746}\\
IIIA - CSIC, Spain\\
levy{@}iiia.scic.es
\and
Mateu Villaret\orcidlink{https://orcid.org/0000-0002-8066-3458}\\
IMA, Universitat de Girona, Spain\\
villaret{@}ima.udg.edu
\and
Yunus Kutz\orcidlink{https://orcid.org/0000-0002-5060-502X} \\
GU Frankfurt, Germany\\
kutz{@}ki.cs.uni-frankfurt.de
}

 \maketitle

 \runninghead{M. Schmidt-Schauss et al.}{Nominal Unification with Recursive Let}

\begin{abstract}
 A sound and complete algorithm for nominal unification of higher-order expressions with a recursive let is described,
 and shown to run in nondeterministic polynomial time. We also explore specializations like nominal letrec-matching for
  expressions, for DAGs, and for garbage-free expressions and determine their complexity.
  We also provide a nominal unification algorithm
 for higher-order expressions with recursive let and atom-variables, where we show that it also runs in
 nondeterministic polynomial time. In addition we prove that there is a guessing strategy for nominal unification with letrec and  atom-variable that is a trade-off
 between exponential
 growth and non-determinism. Nominal matching with variables representing partial letrec-environments is also shown to be in NP.
\end{abstract}

\begin{keywords}
Nominal unification, lambda calculus, higher-order expressions, recursive let,  atom variables 
\end{keywords}

\emph{This paper is an extended version of the conference publication \cite{schmidt-schauss-kutsia-levy-villaret-LOPSTR-16}.}

\ignore{
Manfred: https://orcid.org/0000-0001-8809-7385

Mateu: https://orcid.org/0000-0002-8066-3458

Temur: https://orcid.org/0000-0003-4084-7380  (it seems that there are not contents...)

Yunus: https://orcid.org/0000-0002-5060-502X

Jordi:   https://orcid.org/0000-0001-5883-5746
}

\section{Introduction}

Unification \cite{baadersnyder:2001}  is an operation to make two logical expressions equal by finding substitutions into variables.
There  are numerous applications in computer science, in particular of (efficient) first-order unification, for example in automated reasoning,
type checking and verification.
Unification algorithms are also extended to higher-order calculi with various equivalence relations.
If  equality includes $\alpha$-conversion and $\beta$-reduction and perhaps also $\eta$-conversion
of a (typed or untyped) lambda-calculus,
then unification procedures are known  (see, e.g., \cite{huet:75}), however,
the problem is undecidable \cite{goldfarb:81,levy-veanes:00}.

Our motivation comes from syntactical reasoning on higher-order expressions, with equality being $\alpha$-equivalence of expressions,
and where a unification algorithm is demanded as a basic service.
Nominal unification is the extension of first-order unification with abstractions. It unifies expressions
\wrt $\alpha$-equivalence, and employs permutations
as a mathematically clean treatment of renamings.
It is known that  nominal unification is decidable
\cite{urban-pitts-gabbay:03,DBLP:journals/tcs/UrbanPG04}, where the complexity of the decision problem is
polynomial time \cite{DBLP:journals/tcs/CalvesF08}.
It can be seen also from a higher-order perspective \cite{levy-villaret:12}, as equivalent to
Miller's higher-order pattern unification \cite{DBLP:journals/logcom/Miller91}.
There are efficient algorithms \cite{DBLP:journals/tcs/CalvesF08,levy-villaret:10}, formalizations of nominal unification
\cite{rincon-fernandez-rocha:16}, formalizations with extensions to commutation properties within expressions
\cite{ayala-arvalho-fernandez-nantes:16}, and generalizations of nominal unification to narrowing
\cite{ayala-fernandez-nantes:16}. Equivariant (nominal) unification \cite{cheney-JAR:2010,cheney-diss:2004,aoto-kikuchi:16}
extends nominal unification by permutation-variables, but it can also be seen as a generalization of nominal unification by
permitting abstract names for variables.

We are interested in unification \wrt an additional extension with cyclic let.
To the best of our knowledge, there is no nominal unification algorithm for higher-order expressions permitting general binding  
structures like a cyclic let.  
Higher-order unification could be applied, however, the algorithms are rather general and thus the obtained complexities of specializations
are too high.  Thus we propose to extend and adapt usual nominal unification \cite{urban-pitts-gabbay:03,DBLP:journals/tcs/UrbanPG04} to languages with recursive let.

The motivation and intended application scenario is as follows: constructing  syntactic reasoning algorithms for showing properties
of program transformations on higher-order expressions
in call-by-need functional languages (see for example \cite{moran-sands-carlsson:99,schmidt-schauss-schuetz-sabel:08})
that have a letrec-construct (also called cyclic let) \cite{ariola-klop-short:94} as in Haskell \cite{haskell2010},
(see \eg \cite{cheney:2005} for a discussion on reasoning with more general name binders, and \cite{urban-kaliszyk:12} for a formalization
of general binders in Isabelle).
Extended nominal matching algorithms are necessary for applying program transformations that could be represented as rewrite rules.
Basic properties of  program transformations like commuting properties of conflicting applications or overlaps
can be analyzed in an automated way if there is a nominal unification algorithm of appropriate complexity.
There may be applications also to co-inductive extensions of  logic programming  \cite{simon-mallya-bansal-gupta:06}
and strict functional languages \cite{DBLP:journals/fuin/JeanninKS17a}.
Basically,
overlaps of expressions have to be computed (a variant of critical pairs)
and  reduction steps (under some strategy) have to be performed.
To this end, first an expressive higher-order language is required to represent the meta-notation of expressions.
For example, the meta-notation  $((\lambda x.e_1)~e_2)$ for a beta-redex is made operational
by using unification variables $X_1,X_2$ for $e_1, e_2$.
The scoping of $X_1$ and $X_2$ is different, which can be dealt with by nominal techniques.
In fact,
a more powerful unification algorithm is required for meta-terms employing recursive letrec-environments.

Our main algorithm \letrecunify\ is derived from  first-order unification and nominal unification: From first-order unification
we borrow the decomposition rules,
and the sharing method from Martelli-Montanari-style unification algorithms \cite{martelli-montanari:82}. The adaptations of decomposition
for abstractions and the advantageous use of permutations of atoms is derived from nominal unification algorithms.
Decomposing letrec-expression requires an extension by a permutation of the bindings in the environment, where, however,
one has to take care of scoping. Since in contrast to basic nominal unification, there are nontrivial fixpoints of permutations
(see Example \ref{example:fixpoints-possible}),
novel techniques are required and lead to a surprisingly moderate complexity:
a fixed-point shifting rule (FPS) and a redundancy removing rule (ElimFP) are required. These rules bound the number of
fixpoint equations $X \doteq \pi{\cdot}X$ (where $\pi$ is a permutation) using techniques and results from computations in permutation groups.
The application of these
techniques is indispensable (see Example \ref{example:exponential-FPS}) for obtaining efficiency.

Inspired by the applications in programming languages, we investigate  the notion of garbage-free expressions. The restriction to garbage-free expressions
 permits several optimizations of
 the unification algorithms. The first is that testing $\alpha$-equivalence is polynomial. Another  advantage is that due to the unique correspondence of
 positions for two $\alpha$-equal garbage-free expressions, we show that in this case, fixpoint equations can be replaced by freshness constraints
 (Corollary \ref{cor:no-garbage-no-fixpointeqs}).

As a further extension, we study the possibility to formulate input problems using atom variables as
in \cite{schmidt-schauss-sabel-kutz:19,schmidt-schauss-sabel-fscd:18} in order to take advantage of the potential of less nondeterminism.
The corresponding algorithm {\letrecunifyA} requires permutation expressions and generalizes freshness constraints as further expressibility, and also
other techniques such as explicit compression of permutations. The algorithm runs in NP time.
We added a strategy to really exploit the extended expressivity and the omission of certain nondeterministic choices.

\medskip
\emph{Related Work:} 
Besides the already mentioned related work, we highlight further work. In nominal
commutative unification \cite{DBLP:conf/lopstr/Ayala-RinconSFN17}, one can observe that there are nontrivial fixpoints of permutations. This is similar to what we have in nominal
unification with recursive let (when garbage-freeness is not required), which is not surprising, because, essentially, this phenomenon is related to the lack of the ordering:
in one case among the arguments of a commutative function symbol, in the other case among the bindings of recursive let. Consequently, nominal C-unification reduces to fixpoint
constraints. Those constraints may have infinitely many incomparable solutions expressed in terms of substitutions and freshness constraints (which is the standard way to
represent nominal unifiers). In \cite{DBLP:conf/rta/Ayala-RinconFN18}, the authors proposed to use fixpoint constraints as a primitive notion (instead of freshness constraints)
to axiomatize $\alpha$-equivalence and, hence, use them in the representation of unifiers, which helped to finitely represent solutions of nominal C-unification problems.
The technical report \cite{RISC5292} contains explanations how to obtain a  nominal C-unification algorithm from a letrec unification algorithm and transfers the
NP-completeness result for letrec unification to nominal commutative unification.

An investigation into nominal rewriting and nominal matching is in \cite{Fernandez-Gabbay:07}, where a nominal matching algorithm is
implicitly derived from nominal unification.

 The $\rho_g$-calculus~\cite{DBLP:journals/mscs/BaldanBCK07} integrates term rewriting and lambda calculus, where cyclic, shared terms are permitted.
Such term-graphs are represented as recursion constraints, which resemble to recursive let environments. The evaluation mechanism of the $\rho_g$-calculus is based on matching
for such shared structures. Matching and recursion equations are incorporated in the object level and rules for their evaluation are presented.

Unification of higher-order expressions with recursive let (but without nominal features) has been studied in the context of proving correctness of program transformations
in call-by-need $\lambda$-calculi~\cite{rau-schmidt-schauss:unif:11,rau-schmidt-schauss:unif:10}. Later, in~\cite{schmidt-schauss-sabel:16}, the authors proposed
a more elaborated approach to address semantic properties of program calculi, which involves unification of meta-expressions of higher-order lambda calculi with letrec environments.
This unification problem extends those from~\cite{rau-schmidt-schauss:unif:11,rau-schmidt-schauss:unif:10}: environments are treated as multisets, different kinds of variables
are considered (for letrec environments, contexts, and binding chains), more than one environment variable is permitted, and  non-linear unification problems are allowed.
 Equivalence there is syntactic, in contrast to our nominal approach where equality modulo $\alpha$ is considered. Unlike \cite{schmidt-schauss-sabel:16}, our unification
 problems do not involve context and chain variables,
 but we do have environment variables in matching problems.
 We investigate  an extension of nominal
 letrec unification with atom variables.

\medskip
There are investigations into variants of nominal techniques with a modified view of variables and their renamings and algorithms for the
respective variants of nominal unification
  \cite{DowekG-Gabbay-Mulligan:10}, however, it is unclear whether this can be extended to letrec.

\medskip
{\em Results}: The nominal letrec unification algorithm 
is complete and runs in nondeterministic polynomial time
(Theorem \ref{thm:unification-terminates}, \ref{thm:letrec-unification-in-NP}).
The nominal letrec matching is NP-complete (Theorems \ref{thm:matching-in-NP}, \ref{thm:matching-NP-hard}),
as well as the nominal letrec unification problem
(Theorems \ref{thm:letrec-unification-in-NP}, \ref{thm:matching-NP-hard}).
Nominal letrec matching for DAGs is in NP and outputs substitutions only (subsection \ref{subsec:nom-dag-matching}),
and a very restricted nominal letrec matching problem is already graph-isomorphism hard
(Theorem \ref{thm:matching-GI-hard}).
Nominal unification for garbage-free expressions can be done with simple fixpoint rules (Corollary \ref{cor:no-garbage-no-fixpointeqs}).
In the extension with atom variables, nominal unification can be done using further useful strategies with less nondeterminism
and is NP-complete (Theorem \ref{thm:LRLXA-strategies-complete}).
We construct an algorithm for nominal matching including letrec-environment variables, which runs in NP time (Theorem \ref{thm:letrecenvmatch}).

\medskip
{\em  Structure of the paper.} It starts with a motivating intuition on nominal unification (Sec.  \ref{sec:intuitions}).
After explaining the ground   letrec-language LLR in Sec. \ref{section:language}, the unification algorithm {\letrecunify} for LLR-expression is
described in Sec. \ref{section:algorithm}. Sec. \ref{sec:letrecunify:complete} contains the arguments for soundness and completeness of {\letrecunify}.
Sec. \ref{sec:matching}  describes an improved algorithm for nominal matching on LLR: {\letrecmatch}. Further sections are on extensions.
 Sec.  \ref{sec:letrec-unif-match-hardness} shows Graph-Isomorphism-hardness of nominal letrec matching and unification on
 garbage-free expressions (Theorem \ref{thm:matching-GI-hard}).
Sec. \ref{sec:no-garbage-fixpoints} shows that fixpoint-equations for  garbage-free expressions can be translated into freshness constraints
(Cor. \ref{cor:no-garbage-no-fixpointeqs}). Sec. \ref{section:LRA} considers nominal unification in an extension with atom variables,
 an nominal unification algorithm {\letrecunifyA}  is defined and the differences to  {\letrecunify} are highlighted. It is shown that there is a simple strategy
 such that nominal unification runs in NP time (Theorem \ref{thm:LRLXA-strategies-complete}).
 The last section (Sec.  \ref{sec:atom-letrec-match})  presents
 a nominal matching algorithm {\letrecenvmatch} that is derived from the corresponding nominal unification algorithm   {\letrecunifyA}.
 Sec. \ref{sec:conclusion} concludes the paper.


\section{Some intuitions}\label{sec:intuitions}

In first order unification we have a language of applications of function symbols over a (possible empty) list of arguments $(f e_1\dots e_n)$,
where $n$ is the arity of $f$, and variables $X$. Solutions of equations between terms are substitutions for variables that make both sides of equations
syntactically equal. First order unification problems may be solved using the following two problem transformation rules:\\

(Decomposition) $\gentzen{\Gamma\dotcup \{(f\, e_1 \dots e_n) \doteq (f\, e'_1 \dots e'_n)\}}
       {\Gamma\cup \{e_1\doteq e'_1 \dots e_n\doteq e'_n\}}$\vspace{2mm}\\

(Instantiation) $\gentzen{\Gamma\dotcup \{X \doteq e\}}{[X\mapsto e]\Gamma}$ \begin{minipage}{0.37 \textwidth}\ \ If $X$ does not occur in $e$.\end{minipage}\\

The substitution solving the original set of equation may be easily recovered from the sequence of transformations.
However, the algorithm resulting from these rules is exponential in the worst case.

\medskip
Martelli and Montanari~\cite{martelli-montanari:82} described a set of improved rules that result into an
$O(n\,\log n)$ time algorithm\footnote{The original Martelli and Montanari's algorithm is a bit different. In fact, they do not flatten equations. However, the essence of the algorithm is
basically the same as the one described here.} where $n$ is the size of the input equations. In a first phase the problem is flattened,\footnote{In the flattening process we replace every proper subterm $(f e_1\dots e_n)$ by a fresh variable $X$,
and add the equation $X\doteq (f e_1\dots e_n)$. We repeat this operation (at most a linear number of times) until all proper subterms are variable occurrences.}
resulting into equations where every term is a variable or of the form $(f~X_1 \ldots X_n)$. The second phase is a transformation using the following rules:\\[3mm]
\noindent
 $\begin{array}{@{}ll@{}}
  \begin{minipage}{4cm}(Decomposition) \end{minipage}
   & {\gentzen{\Gamma\dotcup \{(f\, X_1\dots X_n)\doteq (f\, Y_1\dots Y_n)\}}{\Gamma\cup \{X_1\doteq Y_1,\dots,X_n\doteq Y_n\}}} \vspace{3mm} \\

   \begin{minipage}{4cm}  (Variable~ Instantiation) \end{minipage}
    & \gentzen{\Gamma\dotcup \{X \doteq Y\}}{[X\mapsto Y]\Gamma}
 \end{array}
 $
 ~\\[3mm]
\begin{minipage}{4cm}(Elimination)\end{minipage}~ $\gentzen{\Gamma\dotcup \{X \doteq e\}}{\Gamma}$ \begin{minipage}{0.37 \textwidth}\ \ If $X$ neither occurs in $e$ nor in $\Gamma$ \end{minipage}\\

\noindent
(Merge) $\gentzen{\Gamma\dotcup \{X \doteq (f\, X_1\dots X_n) ,X \doteq (f\, Y_1\dots Y_n)\}}{\Gamma\cup \{X \doteq (f\, X_1\dots X_n), X_1\doteq Y_1,\dots,X_n\doteq Y_n\}}$  \\

Notice that in these rules the terms involved in the equations are not modified (they are not instantiated), except by the replacement of a variable
by another in the Variable Instantiation rule. We can define a measure on problems as the number of distinct variables, plus the number of equations,
plus the sum of the arities of the function symbol occurrences. All rules decrease this measure (for instance,
the merge rule increases the number of equations by $n-1$, but removes a function symbol occurrence of arity $n$).
Since this measure is linear in the size of the problem, this proves that the maximal number of rule applications is linear. 
%
The Merge rule is usually described as\\

$\gentzen{\Gamma\dotcup \{X \doteq e_1 ,X \doteq e_2\}}{\Gamma\cup \{X \doteq e_1, e_1\doteq e_2\}}$ \begin{minipage}{0.37 \textwidth}\ \ If $e_1$ and $e_2$ are not variables \end{minipage}  \\

However, this rule does not decrease the proposed measure. We can force the algorithm to, if possible, immediately apply a decomposition of the equation $e_1 \doteq e_2$.
Then, the application of both rules (resulting into the first proposed Merge rule) does decrease the measure.

  \subsection{Nominal unification}

  Nominal unification is an extension of first-order unification in the presence of lambda-binders. Variables of the target language are called atoms,
  and the unification-variables are simply called
  variables.
  Bound atoms can be renamed. For instance,
  $\lambda a.(f\,a)$ is equivalent to $\lambda b.(f\,b)$. We also have permutations of atom names (represented as swappings) applied to
  expressions of the language.
  When these permutations are applied to a variable, this is called a \emph{suspension}. The action of a permutation on a term is simplified until we get a term where permutations
  are innermost and only apply to variables. For instance, $(a\,b) {\cdot} \lambda a.(f\,X\,a\,(f\,b\,c))$, where $(a\,b)$ is a swapping between the atoms $a$ and $b$,
  results into $\lambda b.(f\, (a\,b){\cdot}X\,b\,(f\,a\,c))$.
  As we will see below, we also need a predicate to denote that an atom $a$ cannot occur free in a term $e$, noted $a\# e$.

\medskip
  We can extend the previous first-order unification algorithm to the nominal language modulo $\alpha$-equivalence.
  The decomposition of $\lambda$-expressions distinguishes two cases, when the binder name is the same and when they are distinct and we have to rename one of them:\medskip
	
$\begin{array}{ll@{\qquad}ll} \begin{minipage}{2.5cm} \centering (Decomposition \\ lambda 1) \end{minipage}
 &  \gentzen{\Gamma\dotcup\{\lambda a.s \doteq \lambda a.t\}}
	{\Gamma \cup \{s \doteq t\}}
	 &
 \begin{minipage}{2.5cm} \centering (Decomposition \\ lambda 2) \end{minipage} & \gentzen{\Gamma\dotcup\{\lambda a.s \doteq \lambda b.t\}}
	    {\Gamma \cup \{s \doteq (a~b){\cdot}t, a {\freshdot} t\}}
\end{array}
$

\medskip
As we see in the second rule, we introduce a \emph{freshness
  constraint} that has to be checked or solved, so we need a set of
transformations for this kind of equations.  This set of freshness
constraints is solved in a second phase of the algorithm.

\medskip
As we have said, permutations applied to variables cannot be longer
simplified and result into suspensions. Therefore, now, we deal with suspensions instead of variables, and we do not make any distinction between $X$ and $\mathit{Id}{\cdot} X$.
Variable instantiation distinguishes two cases:\\[2mm]
\noindent
$\begin{array}{@{}ll@{\qquad}ll}
\begin{minipage}{2.5cm} (Variable \\Instantiation) \end{minipage}
& \gentzen{\Gamma\dotcup \{\pi\cdot X \doteq \pi'\cdot Y\}~~~ X \not= Y}{[X\mapsto (\pi^{-1}\circ\pi')\cdot Y]\Gamma}
& \begin{minipage}{1.5cm} (Fixpoint) \end{minipage}
& \gentzen{\Gamma\dotcup \{\pi\cdot X\doteq \pi'\cdot X\}}{\Gamma\cup \{a \# X\mid a\in \dom(\pi^{-1}\circ\pi')\}}
\end{array}
$
 ~\\
 Notice that equations between the same variable $X\doteq X$ that are
trivially solvable in first-order unification, adopt now the form $\pi\cdot X\doteq
\pi'\cdot X$.  This kind of equations are called \emph{fixpoint
  equations} and impose a restriction on the possible instantiations
of $X$, when $\pi$ and $\pi'$ are not the identity.  Namely, $\pi\cdot X \doteq \pi'\cdot X$
is equivalent to $\{a\# X\mid a\in\dom(\pi^{-1}\circ\pi')\}$, where the domain $\dom(\pi)$ is the
set of atoms $a$ such that $\pi(a)\neq a$.

\medskip
From this set of rules we can derive an $O(n^2\,\log n)$ algorithm, similar to the algorithms described in \cite{DBLP:journals/tcs/CalvesF08,levy-villaret:10}.
This algorithm has three phases. First, it flattens all equations. Second, it applies this set of problem transformation rules. Using the same measure as in the first-order
case (considering lambda abstraction as a unary function symbol and not counting the number of freshness equations),
we can prove that the length of problem transformation sequences is always linear. In a third phase, we deal with freshness equations.
Notice that the number of
distinct non-simplifiable freshness equations $a\#X$ is quadratically bounded.

\subsection{Letrec expressions}

  Letrec expressions have the form $(\tletr~a_1.e_1; \ldots; a_n.e_n~\tin~e)$. Variables $a_i$ are binders
  where the scope is
  in all expressions $e_j$ and in $e$. (Here we will use $\alpha$-equivalence in an informal fashion; it is defined in Def. \ref{def:alpha-equivalence-sim}.)
  We view the environment part $a_1.e_1; \ldots; a_n.e_n$ as a multiset. We can rename these binders, obtaining an equivalent expression.
  For instance, $(\tletr~a.(f\,a)~\tin~(g\,a)) \sim (\tletr~b.(f\,b)~\tin~(g\,b))$ \footnote{Here we will use $\sim$ informally as $\alpha$-equivalence, which will be formally
  defined in Def. \ref{def:alpha-equivalence-sim}}.
  Moreover, we can also swap the order of definitions. For instance, $(\tletr~a.f; b.g~\tin~(h\,a\,b)) \sim (\tletr~b.g; a.f~\tin~(h\,a\,b))$.
  Schmidt-Schau\ss\ et al.~\cite{schmidtschauss-sabel-rau-RTA:2013} prove that equivalence of letrec expressions is graph-isomorphism (GI) complete and
  Schmidt-Schau\ss\ and Sabel~\cite{schmidt-schauss-sabel:16} prove that unification is NP-complete.
  The GI-hardness can be elegantly proved by encoding any graph, like $G=(V,E)=(\{v_1,v_2,v_3\},\{(v_1,v_2),(v_2,v_3)\})$, into a letrec expression,
  like $(\tletr~v_1.a; v_2.a; v_3.a$ $\tin~\tletr~e_1.(c\, v_1\, v_2);e_2.(c\, v_2\, v_3)~\tin~a)$.
  Here, $v_i$ represent the nodes and $(c\, v_i\, v_j)$ the edges of the graph.

   \medskip
   There are nontrivial fixpoints of permutations in the letrec-language.
    For example, $(\tletr~a_1.b_1,$ $a_2.b_2, a_3.a_3~\tin~ a_3)$ is a fixpoint of the equation $X\doteq (b_1~b_2)\cdot X$, although
   $b_1$ and $b_2$ are not fresh in the expression, which means $(b_1~b_2){\cdot}(\tletr~a_1.b_1,$ $a_2.b_2, a_3.a_3~\tin~ a_3)$ $\sim$
   $(\tletr~a_1.b_1,$ $a_2.b_2, a_3.a_3~\tin~ a_3)$.
   Therefore, the fixpoint rule of the  nominal algorithm in \cite{urban-pitts-gabbay:03}
   would not be complete in our setting:
   to ensure $X\doteq (b_1~b_2)\cdot X$ we cannot
   require $b_1\# X$ and $b_2\# X$. See also Example~\ref{example:fixpoints-possible}. Hence, fixpoint equations can in general not be replaced by freshness constraints.
   For the general case we need a complex elimination rule, called fixed point shift:\\

   \noindent (FixPointShift) $\gentzen{\Gamma\dotcup \{\pi_1{\cdot}X \doteq \pi'_1{\cdot}X, \ldots,\pi_n{\cdot}X \doteq \pi'_n{\cdot}X, \pi{\cdot}X \doteq e\}}
	{\Gamma \cup  \{\pi_1\pi^{-1}{\cdot}e \doteq \pi'_1\pi^{-1}{\cdot}e, \ldots,\pi_n\pi^{-1}{\cdot}e \doteq \pi'_n\pi^{-1}{\cdot}e\}}$,
%
%
   \hspace*{3mm} \begin{minipage}{0.25 \textwidth} if $X$ neither occurs in $e$ nor in $\Gamma$. \end{minipage}\\

   The substitution is $X \to \pi^{-1}{\cdot}e$.
   This rule can generate an exponential number of equations (see Example~\ref{example:exponential-FPS}). In order to avoid this effect, we will use a property on the number of
   generators of permutation groups (see end of Section~\ref{section:language}).

\medskip
   For the decomposition of letrec expressions we also need to introduce a (don't know) nondeterministic choice. \\[3mm]
  $\gentzen{\Gamma \dotcup\{\tletr ~a_1.s_1;\ldots;a_n.s_n~\tin~r \doteq \tletr ~b_1.t_1;\ldots;b_n.t_n~\tin~r' \} }
         {\mid_{\{\rho\}} (\Gamma \, \cup \{
                   s_1 \doteq \pi{\cdot}t_{\rho(1)},\ldots,s_n \doteq \pi{\cdot}t_{\rho(n)}, r \doteq \pi{\cdot}r'\}}$\\

\smallskip
\noindent   Where the  necessary freshness constraints are $\{a_i \# (\tletr ~b_1.t_1;\ldots;b_n.t_n~\tin~r') \mid i = 1,\ldots,n\}$.
                           The permutation $\rho$  on $\{1,\ldots,n\}$ is chosen using don't-know non-determinism, indicated by the vertical bar and $\{\rho\}$,
                  and $\pi$ is an (atom-)permutation  that extends
               $\{b_{\rho(i)} \mapsto a_i \mid i = 1,\ldots,n)\}$   with $\dom(\pi) \subseteq \{a_1,\ldots,a_n,b_1,\ldots,b_n\}$.

      \smallskip\noindent In Section~\ref{section:algorithm}, we will describe in full detail all the transformation rules of our algorithm.

   \section{The ground language of expressions}\label{section:language}

The very first idea of nominal techniques \cite{urban-pitts-gabbay:03} is to use concrete variable names in  lambda-calculi (also in extensions),
in order to avoid implicit $\alpha$-renamings, and instead use operations for explicitly applying bijective renamings. 
Suppose $s = \lambda {\tt x}.{\tt x}$ and $t = \lambda {\tt y}.{\tt y}$ are concrete (syntactically different) lambda-expressions.
The nominal technique provides explicit name-changes using permutations.
These permutations are applied irrespective of binders. For example $({\tt x}~{\tt y}){\cdot}(\lambda {\tt x}. \lambda {\tt x}.{\tt a})$ results in
$\lambda {\tt y}. \lambda {\tt y}.{\tt a}$.
Syntactic reasoning on higher-order expressions, for example unification of higher-order expressions modulo $\alpha$-equivalence
will be emulated by nominal
techniques on a language with concrete names, where the algorithms require certain extra constraints and operations.
The gain is that all conditions
and substitutions etc. can be computed and thus more reasoning tasks can be automated,
whereas the implicit name conditions under non-bijective renamings 
have
a tendency to complicate (unification-) algorithms and to hide the required conditions on equality/disequality/occurrence/non-occurrence of names.
We will stick to a notation closer to lambda calculi than most other papers on nominal unification, however, note that in general the differences
are only in notation and the constructs like application and abstraction can easily be translated into something equivalent in the other language
 without any loss.

\subsection{Preliminaries}
We define the language $\LRL$ ({\bf L}et{\bf R}ec {\bf L}anguage) of (ground-)expressions, which is a lambda calculus extended with
a recursive let construct. The notation is consistent with  \cite{urban-pitts-gabbay:03}.
The (infinite) set of atoms $\MBA$ is a set of (concrete) symbols $a,b$ which we usually denote in a meta-fashion; so we can use
symbols $a,b$  also with indices (the variables in lambda-calculus).
There is a set ${\cal F}$ of function symbols with arity $\ari(\cdot)$.
The syntax of the expressions $e$ of $\LRL$ is:
\[ e   ::=   a \mid \lambda a.e \mid (f~e_1~\ldots ~e_{\ari(f)}) \mid (\tletr~a_1.e_1; \ldots; a_n.e_n~\tin~e)
\]

We assume that binding atoms $a_1,\ldots,a_n$ in a letrec-expression $(\tletr~a_1.e_1; \ldots;$  $a_n.e_n~\tin~e)$   are pairwise distinct.
Sequences of bindings $a_1.e_1;\ldots; a_n.e_n~$ may be abbreviated as $\env$.
The expressions $(\tletr~a_1.e_1; \ldots; a_n.e_n~\tin~e)$ and $(\tletr~a_{\rho(1)}.e_{\rho(1)}; \ldots; a_{\rho(n)}.e_{\rho(n)}~\tin~e)$
are defined as equivalent for every permutation $\rho$ of $\{1,\ldots,n\}$, i.e. in the following we view the environment $a_1.e_1; \ldots; a_n.e_n$ of a letrec-expression
as a multi-set.  

\medskip
The {\em scope} of atom $a$ in $\lambda a.e$  is standard: $a$ has scope $e$.
The $\tletr$-construct  has a special scoping rule:
in $(\tletr ~a_1.s_1; \ldots;a_n.s_n~\tin~r)$, every atom $a_i$ that is free in some $s_j$  or $r$ is bound by the environment
$a_1.s_1; \ldots;a_n.s_n$.
This defines in $\LRL$  the notion of free atoms
$\FA(e)$, bound atoms $\BA(e)$ in expression $e$, and all atoms $\AT(e)$ in $e$.
For an environment $\env = \{a_1.e_1,\ldots,a_n.e_n\}$, we define the set of letrec-atoms as $\LA(\env) = \{a_1,\ldots,a_n\}$.
Note that this is well-defined, since environments are multisets, but names are meant syntactically.  
We say {\em $a$ is fresh for $e$} iff  $a \not\in \FA(e)$ (also denoted as $a\#e$).
As an example, the expression $(\tletr~a.\mathit{cons}~s_1~b; b.cons~ s_2~ a~\tin ~a)$ represents an infinite list      
$(cons~s_1~(cons~s_2~(cons~s_1 ~(cons~s_2~\ldots))))$, where $s_1,s_2$ are expressions.
The functional application operator in functional languages (which is usally implicit) can be encoded by a binary function {\tt app},
which also allows to deal with partial applications.
Our language $\LRL$ is
a fragment of
core calculi  \cite{moran-sands-carlsson:99,schmidt-schauss-schuetz-sabel:08},
since for example the case-construct is missing, but this could also be represented.

We will use mappings on atoms from $\MBA$. A {\em swapping} $(a~b)$ is a bijective function (on $\LRL$-expressions) that maps an atom $a$ to atom $b$, atom $b$ to $a$,
and is the identity on other atoms.
We will also use finite permutations $\pi$ on atoms from $\MBA$, which could be
represented as a composition of swappings in the algorithms below.
Let $\dom(\pi) = \{a \in \MBA \mid \pi(a) \not= a\}$. Then every finite permutation can be represented by a composition of
at most $(|\dom(\pi)|- 1)$ swappings.
Composition  $\pi_1 \circ \pi_2$ and inverse $\pi^{-1}$ can be immediately computed, where the complexity is polynomial in the size of $\dom(\pi)$.
Permutations $\pi$ operate on expressions simply by recursing on the structure.
For a  letrec-expression this is
$\pi\cdot(\tletr ~a_1.s_1; \ldots;a_n.s_n~\tin~e)$ $=$
$(\tletr ~\pi\cdot a_1.\pi\cdot s_1; \ldots;\pi\cdot a_n.\pi\cdot s_n~\tin~\pi\cdot e)$.
Note that permutations also change names of bound atoms.  

\medskip We will use the following definition (characterization) of $\alpha$-equivalence:
\begin{definition}\label{def:alpha-equivalence-sim}
	The $\alpha$-equivalence $\sim$ on expressions $e \in \LRL$ is defined as follows:
	\begin{itemize}
		\item $a \sim a$.
		\item if $e_i \sim e_i'$ for all $i$, then $(f e_1 \ldots e_n) \sim (f e_1' \ldots e_n')$ for an $n$-ary  $f \in {\cal F}$.
		\item If $e \sim e'$, then $\lambda a.e \sim \lambda a.e'$.
		\item If $a\#e'$ and $e \sim (a~b)\cdot e'$, then $\lambda a.e \sim \lambda b.e'$.    
       \item If there is a permutation $\pi$ on atoms
		 such that
		   \begin{itemize}
		     \item $\dom(\pi) \subseteq  \{a_1,\ldots,a_n\} \cup \{b_1,\ldots,b_n\}$, where $a_i \not= a_j$ and  $b_i \not= b_j$ for all $i \not= j$,
		      \item  $\pi(b_i) = a_i$ for all $i$,
              \item $\{a_1,\ldots,a_n\} \#  (\tletr~b_1.t_1, \ldots,b_n.t_n~\tin ~r')$,  and
              \item $r \sim \pi(r')$ and   $s_i \sim \pi(t_i)$      for $i = 1,\ldots,n$ hold.
            \end{itemize}
       Then $(\tletr~a_1.s_1, \ldots,a_n.s_n~\tin ~r) \sim (\tletr~b_1.t_1, \ldots,b_n.t_n~\tin ~r')$.
	\end{itemize}
\end{definition}
The last phrase includes   $(\tletr~a_1.s_1, \ldots,a_n.s_n~\tin ~r) \sim (\tletr~b_{\rho(1)}.t_{\rho(1)}, \ldots,b_{\rho(n)}.t_{\rho(n)}$~$\tin$ $~r')$
for every permutation $\rho$ on $\{1,\ldots,n\}$, by the definition of syntactic equality that treats the let-environment as a multi-set.

\medskip
Note that $\{a_1,\ldots,a_n\} \#  (\tletr~b_1.t_1, \ldots,b_n.t_n~\tin ~r')$ is equivalent to\\
$(\{a_1,\ldots,a_n\} \setminus \{b_1,\ldots,b_n\})$ $\#$  $(\tletr~b_1.t_1, \ldots,b_n.t_n~\tin ~r')$.
Note also that $\sim$ is identical to $\alpha$-equivalence, i.e., the relation generated by renamings of binding constructs and permutation of bindings
in a letrec. We omit a proof, since it detracts the attention from the main contents. Such a proof is not hard to construct by using that $\alpha$-equivalence holds, if
and only if the graph constructed by replacing bindings by pointing edges, where the  outgoing edges from an environment are unordered and the one from a function
 application are ordered.
 Note that our view is that algorithms work on the syntactic terms as given, and not with equivalence classes modulo $\sim$. 

A nice and  important property that is often implicitly used is:   $e_1 \sim e_2$ is equivalent to $\pi{\cdot} e_1 \sim \pi{\cdot} e_2$ for any (atom-)permutation $\pi$.

In usual nominal unification, the solutions of fixpoint equations $X \doteq \pi{\cdot}X$, i.e. the  sets $\{e \mid \pi \cdot e \sim e\}$
 can be characterized by using finitely many freshness constraints  \cite{urban-pitts-gabbay:03}.
Clearly, all these sets and also all finite intersections are nonempty,
since at least fresh atoms are elements and since $\MBA$ is infinite.
However, in our setting, these sets are nontrivial:
\begin{example}\label{example:fixpoints-possible}   The $\alpha$-equivalence  $(a~b)\cdot (\tletr~c.a;d.b~\tin ~\mathit{True})$ $\sim$
	$(\tletr~ \allowbreak c.a;d.b ~\tin ~\mathit{True})$ holds, which means that there
	are expressions $t$ in $\LRL$ with  $t \sim (a~b)\cdot t$ and  $\FA(t) = \{a,b\}$. This is in contrast to usual nominal unification.
\end{example}

\subsection{Permutation  groups}\label{subsec:permutation-groups}

Below we will use the results on complexity of operations in finite permutation groups, see
\cite{luks:91,Furst-Hopcroft-Luks:80}. We summarize some facts on the so-called symmetric group and its properties.
We consider a set $\{o_1,\ldots,o_n\}$ of distinct objects $o_i$
(in our case atoms),
and the symmetric group $\Sigma(\{o_1,\ldots,o_n\})$ (of size $n!$) of permutations of these objects.
We will also look at its elements, subsets and subgroups. Subgroups of $\Sigma(\{o_1,\ldots,o_n\})$
can always be represented by a set of generators (represented as permutations on $\{o_1,\ldots,o_n\}$).
If $H$ is a set of elements (or generators), then $\gen{H}$ denotes the generated subgroup of $\Sigma(\{o_1,\ldots,o_n\})$.
Some facts are:
\begin{itemize}
\itemsep=0.9pt
	\item A permutation can be represented in space linear in $n$.
	\item Every subgroup of $\Sigma(\{o_1,\ldots,o_n\})$ can be represented by   $\le n^2$ generators.
\end{itemize}
However, elements in a subgroup may not be representable as a product of polynomially many of these generators.

\medskip
The following questions can be answered in polynomial time:
\begin{itemize}
\itemsep=0.9pt
	\item The element-question: $\pi \in G$.
	\item The subgroup question: $G_1 \subseteq G_2$.
\end{itemize}

However, intersection of groups and set-stabilizer
(i.e. $\{\pi \in G \mid \pi(M) = M\}$)
are not known to be computable in polynomial time, since those problems are as hard as graph-isomorphism (see  \cite{luks:91}).\vspace*{-1mm}

\section{A nominal letrec unification algorithm}\label{section:algorithm}
 \subsection{Preparations} \label{subsec:unifalg-preparations}

As an extension of  $\LRL$, there is a countably infinite  set of (unification) variables $\Var$ ranged over by $X,Y$ where we also use indices.
The syntax of the language $\LRLX$ ({\bf L}et{\bf R}ec {\bf L}anguage e{\bf X}tended) is
\[\begin{array}{lcl}
e & ::=&  a \mid X \mid \pi\cdot{}X \mid \lambda a.e \mid (f~e_1~\ldots e_{\ari(f)})~| ~(\tletr~a_1.e_1; \ldots; a_n.e_n~\tin~e) \\
\pi &:=& \emptyset \mid (a~b){\cdot}\pi
\end{array}
\]
$\Var(e)$ is the set of variables $X$ occurring in $e$.

\begin{figure}[!t]\small
	\begin{flushleft} $\inferrule*[right=\mbox{\normalfont  if $a \neq b$}]{\{a{\freshdot}b\} \dotcup \nabla}{\nabla}$ \quad   
		$\inferrule*{\{a{\freshdot}(f~s_1 \ldots s_n)\} \dotcup \nabla}{\{a{\freshdot}s_1,\ldots,a{\freshdot}s_n\} \cup \nabla}$  \quad
		$\inferrule*{\{a{\freshdot}(\lambda a.s)\} \dotcup \nabla}{\nabla}$ \quad
		$\inferrule*[right=\text{\normalfont  if $a \neq b$}]{\{a{\freshdot}(\lambda b.s)\} \dotcup \nabla}{\{a{\freshdot}s\} \cup \nabla}$\\[2mm]  
		$\inferrule*[right=\text{\normalfont \mbox{if $a \in \{a_1,\ldots,a_n\}$}}]
		{\{a{\freshdot}(\tletr~a_1.s_1;\ldots,a_n.s_n~\tin~r)\} \dotcup \nabla}{\nabla}$ \hspace*{1cm}  
		$\inferrule*{\{a{\freshdot}a\} \dotcup \nabla}{\bot}$  \quad
		\\[1mm]
		$\inferrule*[right=\text{\normalfont  \mbox{if $a \not\in \{a_1,\ldots,a_n\}$}}]
		{\{a{\freshdot}(\tletr~a_1.s_1;\ldots,a_n.s_n~\tin~r)\} \dotcup \nabla}
		{\{a{\freshdot}s_1,\ldots a{\freshdot}s_n,a{\freshdot}r\} \cup \nabla}$  ~~~~~~
		$\inferrule*{\{a{\freshdot}(\pi\cdot X)\} \dotcup \nabla}{\{\pi^{-1}(a) {\freshdot} X\} \cup \nabla}$  
	\end{flushleft}\vspace*{-4mm} \normalsize
	\caption{Simplification of freshness constraints in $\LRLX$}\label{fig:def:simplification-freshness}
  \end{figure}

\medskip
The expression $\pi{\cdot}e$ for a non-variable $e$ means an operation, which is performed by shifting $\pi$ down,
using the additional simplification $\pi_1 {\cdot} (\pi_2 {\cdot} e)$ $\to $ $(\pi_1 \circ  \pi_2) {\cdot} e$, 
where after the shift, $\pi$ only occurs in the subexpressions of the form $\pi\cdot{}X$, which are called {\em suspensions}.
Usually, we do not distinguish $X$  and $\mathit{Id}{\cdot}X$, notationally.
A single {\em freshness constraint} in our unification algorithm is of the form $a{\freshdot}e$, where $e$ is an $\LRLX$-expression, and an
{\em atomic} freshness constraint is of the form $a{\freshdot}X$.
A conjunction (or set) of freshness constraints is sometimes called {\em freshness context}.  

\begin{lemma}\label{lemma:simplify-freshness-constraint}
The rules in Fig. \ref{fig:def:simplification-freshness} for simplifying sets of freshness constraints in $\LRLX$ run in polynomial time and the result is
either $\bot$, i.e. fail, or a set of freshness constraints
 where all single constraints are atomic.   
 This constitutes  a polynomial decision algorithm for satisfiability of $\nabla$: If $\bot$ is in the result, then unsatisfiable,
 otherwise satisfiable.
\end{lemma}
We can assume in the following algorithms that sets of freshness constraint are immediately simplified.
In the following we will use $\Var(\Gamma, \nabla)$, and  $\Var(\Gamma, e)$  and similar notation for the set of unification-variables occurring in the syntactic
  objects mentioned in the brackets.

\begin{definition}
    An {\em $\LRLX$-unification problem} is a pair $(\Gamma, \nabla)$, where  $\Gamma$ is a set of equations $\{s_1 \doteq t_1, \ldots,s_n \doteq t_n\}$,   
    and $\nabla$ is a set of freshness constraints $\{a_1\freshdot X_1,\ldots,a_m\freshdot X_m\}$.
        \noindent A {\em (ground) solution} of $(\Gamma, \nabla)$  is a substitution $\rho$, mapping variables in $\Var(\Gamma, \nabla)$ to ground expressions,
	such that $s_i\rho \sim t_i\rho$, for $i = 1,\ldots,n$,
	and  $a_j \# (X_j\rho)$, for $j=1,\ldots,m$. \\
      The decision problem is whether there is a ground solution for a given $(\Gamma, \nabla)$ or not.
\end{definition}

For the unification algorithms below, we employ a representation of unifiers as iterated single substitutions which is like a DAG-compression of a substitution.
For example the representation $\{x \mapsto f(y,z), y \mapsto f(a,a), z \mapsto g(b,b)\}$ represents the substitution $\{x \mapsto f(f(a,a),g(b,b)), y \mapsto f(a,a), z \mapsto g(b,b) \}$.

\begin{definition}\label{def:LRLX-unifier}
	Let  $(\Gamma, \nabla)$ be an $\LRLX$-unification problem. We consider triples $(\sigma,\nabla',\FIX)$ as representing general unifiers,
	where $\sigma$ is
	a substitution represented by a sequence of single assignments (which has the effect of a DAG-compression), 
	mapping variables to $\LRLX$-expressions,
	$\nabla'$ is a set of freshness constraints, and  $\FIX$ is a set of fixpoint equations of the form $\pi'{\cdot}X \doteq \pi{\cdot}X$,
	where $X \not\in \dom(\sigma)$.

\medskip
	\noindent A triple $(\sigma,\nabla',\FIX)$ is a {\em unifier} of $(\Gamma, \nabla)$, if
    \begin{itemize}
	 \item[(i)] there exists a ground substitution $\rho$ that solves $(\nabla'\sigma,\FIX)$, i.e.,
	    for every $a{\#}X$ in $\nabla'$, $a{\#}X\sigma\rho$ is valid, and for every fixpoint equation $\pi'{\cdot}X \doteq \pi{\cdot}X \in \FIX$, it holds
	    $\pi'{\cdot}(X\rho) \sim \pi{\cdot}(X\rho)$; and
	 \item[(ii)] for every ground substitution $\rho$ that instantiates all variables in $Var(\Gamma, \nabla)$ and
	   which solves $(\nabla'\sigma,\FIX)$,
	   the ground substitution $\sigma\rho$ is a solution of  $(\Gamma, \nabla)$.
    \end{itemize}
	\noindent A set $M$ of unifiers is {\em complete}, if every solution $\mu$ is covered by at least one unifier,
	i.e., there is some unifier $(\sigma,\nabla',\FIX)$ in $M$, and a ground substitution $\rho$,
	such that $X\mu \sim X\sigma\rho$ for all $X \in \Var(\Gamma, \nabla)$.
\end{definition}

We will employ nondeterministic rule-based algorithms computing unifiers: There are  clearly indicated
disjunctive (don't know nondeterministic) rules, and all other rules are don't care nondeterministic.
This distinction is related to the completeness of the solution algorithms: don't care  means that the rule has several possibilities
where it is sufficient for completeness to take only one. On the other hand, don't know means that for achieving completeness, every possibility
of the rule has to be explored.
The {\em collecting variant} of the algorithm  runs  and collects all solutions from all alternatives of the disjunctive rule(s).
The {\em decision variant} guesses and verifies one possibility and tries to detect the existence of 
a single unifier.

Since we want to avoid the exponential size explosion of the Robinson-style unification, keeping  the good properties
of Martelli Montanari-style algorithms \cite{martelli-montanari:82}, 
we stick to  a set of equations as data structure.
As a preparation for the algorithm,
all expressions in equations are exhaustively flattened as follows:
$(f~t_1 \ldots  t_n) \to (f~X_1 \ldots X_n)$ plus the equations $X_1 \doteq t_1,\ldots,X_n \doteq t_n$.
Also $\lambda a.s$ is replaced by $\lambda a.X$  with equation $X \doteq s$,
and $(\tletr ~a_1.s_1;\ldots,a_n.s_n~\tin~r)$ is replaced by  $(\tletr~ a_1.X_1;\ldots,a_n.X_n~\tin~X)$ with the additional equations
$X_1 \doteq s_1; \ldots ;X_n \doteq s_n;X \doteq r$. The introduced variables $X_i,X$ are fresh ones.
%
Thus, all expressions in equations are of depth at most 1, not counting the permutation
applications in the suspensions.

In the notation of the rules, we use $[e/X]$  as substitution that replaces $X$ by $e$, whereas $\{X \to t\}$ is used for constructing a syntactically represented substitution.
  In the written rules, we  may omit $\nabla$ or $\theta$ if they are not changed.
We will use a  notation ``$|$'' in the consequence part of a rule, usually with a set of possibilities indicated by for example $\{\rho\}$, to denote
disjunctive  (i.e. don't know) nondeterminism.  The only nondeterministic rule that requires exploring all alternatives is rule (6) in Fig. \ref{lrunify-rules-1}.
The other rules can be applied in any order, where it is not necessary to explore alternatives.  

\begin{figure}[!h]\small
\vspace{1mm}
\begin{flushleft}
	$(1)~\gentzen{\Gamma \dotcup\{e \doteq e\},\nabla,\theta}{\Gamma,\nabla,\theta}$  \qquad
	$(2)~\gentzen{\Gamma \dotcup\{\pi_1{\cdot}X \doteq \pi_2{\cdot}Y\},\nabla, \theta \qquad Y \not= X}{\Gamma[\pi_1^{-1}\pi_2{\cdot}Y/X],
	 \nabla[\pi_1^{-1}\pi_2{\cdot}Y/X], \theta \cup \{X \mapsto \pi_1^{-1}\pi_2{\cdot}Y\}}$   \\[2mm]
	%
	$(3)~ \gentzen{\Gamma\dotcup\{(f~s_1 \ldots s_n) \doteq (f~s_1' \ldots s_n')\},\nabla,\theta}   
	{\Gamma \cup \{s_1 \doteq s_1',\ldots,s_n \doteq s_n'\},\nabla,\theta}$\\[2mm]        

	$(4)~ \gentzen{\Gamma\dotcup\{\lambda a.s \doteq \lambda a.t\},\nabla,\theta}   
	{\Gamma \cup \{s \doteq t\},\nabla,\theta}$    
	\qquad 
	$(5)~ \gentzen{\Gamma\dotcup\{\lambda a.s \doteq \lambda b.t\},\nabla,\theta  \qquad   a \not= b}
	{\Gamma \cup \{s \doteq (a~b){\cdot}t\} ,\nabla \cup  \{a {\freshdot} t\},\theta}$
	\\[2mm]

$(6)~ \gentzen{\Gamma \dotcup\{\tletr ~a_1.s_1;\ldots;a_n.s_n~\tin~r \doteq \tletr ~b_1.t_1;\ldots;b_n.t_n~\tin~r' \}, \nabla,\theta}  
         {\begin{array}{l}
           \begin{array}{ll}
         \hspace*{-1mm}    
           {\begin{array}{c}
               \left|\begin{array}{c}
                ~\\
               \end{array}\right.
               \\
           {\hspace*{-3mm}\{\rho\}}   
           \end{array}}
            &
              \begin{array}{l}
                (\Gamma \, \cup \{
                   s_1 \doteq \pi{\cdot}t_{\rho(1)},\ldots,s_n \doteq \pi{\cdot}t_{\rho(n)}, r \doteq \pi{\cdot}r'\}, \\
              ~~~\nabla \cup \{a_i \# (\tletr ~b_1.t_1;\ldots;b_n.t_n~\tin~r') \mid i = 1,\ldots,n\},\theta)
               \end{array}
             \end{array}
                  \\[2mm]
              \mbox{where $\rho$ is a  permutation on $\{1,\ldots,n\}$. The permutation $\pi$ is chosen don't care such that}\\
               \mbox{ it extends $\{b_{\rho(i)} \mapsto a_i \mid i = 1,\ldots,n)\}$   with $\dom(\pi) \subseteq \{a_1,\ldots,a_n,b_1,\ldots,b_n\}$ }   \\
      \end{array}
}
      $
\end{flushleft}\vspace*{-5mm}
\caption{Standard (1,2) and decomposition rules (3,4,5,6)}\label{lrunify-rules-1}\vspace*{-5mm}
\end{figure}

\subsection{Rules of the algorithm \letrecunify}

\begin{figure}[t]\small
\begin{flushleft}
 \mbox{(MMS)} $\gentzen{\Gamma\dotcup \{\pi_1{\cdot}X \doteq e_1, \pi_2{\cdot}X \doteq e_2\},\nabla,\theta}
	{\Gamma \cup  \{\pi_1{\cdot}X \doteq e_1\} \cup \Gamma', \nabla \cup \nabla',\theta}$,~~~
	\begin{minipage}{0.50 \textwidth} if  $e_1, e_2$ are not suspensions,
	  where $\Gamma'$ is the set of equations generated by decomposing $\pi_1^{-1}{\cdot}e_1 \doteq \pi_2^{-1}{\cdot}e_2$ using (3)--(6),
	  and where $\nabla'$ is the corresponding resulting set of freshness constraints. \\
	\end{minipage}
	\\[-1mm]
	
	\mbox{(FPS)}~$\gentzen{\Gamma\dotcup \{\pi_1{\cdot}X \doteq \pi'_1{\cdot}X, \ldots,\pi_n{\cdot}X \doteq \pi'_n{\cdot}X, \pi{\cdot}X \doteq e\}, \nabla, \theta}
	{\Gamma \cup  \{\pi_1\pi^{-1}{\cdot}e \doteq \pi'_1\pi^{-1}{\cdot}e, \ldots,\pi_n\pi^{-1}{\cdot}e \doteq \pi'_n\pi^{-1}{\cdot}e\},  \nabla, \theta \cup \{X \mapsto \pi^{-1}{\cdot}e\}}$, \\
	\hspace*{1.1cm} \begin{minipage}{0.9 \textwidth}  
		If neither $X \in \Var(\Gamma,e)$,
		nor $e$ is a suspension, nor (Cycle) (see Fig.\ref{def:failure-rules}) is applicable.
	     \end{minipage}
	\\[2mm]
	
	\mbox{(ElimFP)}~~$\gentzen{\Gamma\dotcup \{\pi_1{\cdot}X \doteq \pi'_1{\cdot}X, \ldots,\pi_n{\cdot}X \doteq \pi'_n{\cdot}X, \pi{\cdot}X \doteq \pi'{\cdot}X\}, \nabla, \theta}
	{\Gamma\cup \{\pi_1{\cdot}X \doteq \pi'_1{\cdot}X, \ldots,\pi_n{\cdot}X \doteq \pi'_n{\cdot}X\}, \nabla, \theta}$,   \\
	\hspace*{1.5cm} \text{If } $\pi^{-1}\pi' \in \gen{\pi_1^{-1}\pi_1',\ldots,\pi_n^{-1}\pi_n'}.$
	\\[2mm]
	
	\mbox{(Output)}~~$\gentzen{\Gamma, \nabla, \theta}
	{(\theta, \nabla, \Gamma)}$     
	\begin{minipage}{0.45\textwidth} if $\Gamma$ only consists of fixpoint-equations.  \end{minipage}
\end{flushleft}\vspace*{-5mm}
\caption{ Main Rules of \letrecunify}\label{lrunify-rules-2}
\end{figure}

The top symbol of an expression is defined as 
$\tops(f~s_1 \ldots s_n) = f$,  $\tops(a) = a$,
$\tops(\lambda a.s) = \lambda$,  and
$\tops(\tletr~\env~\tin~s) = (\tletr,n)$, where $n$ is the number of bindings in $env$. It is undefined for variables $X$.

\begin{definition} The rule-based algorithm  {\letrecunify}  is defined in the following. Its rules are in Figs. \ref{lrunify-rules-1}, \ref{lrunify-rules-2} and
\ref{lrunify-rules-3}.
\letrecunify\  operates on a tuple $(\Gamma, \nabla, \theta)$, where
$\Gamma$ is a set of flattened equations $e_1 \doteq e_2$, and where we assume that  $\doteq$ is symmetric,
$\nabla$ contains freshness constraints,
and $\theta$ represents the already computed substitution as a list of mappings of the form $X \mapsto e$. Initially $\theta$  is empty.

The final state will be reached, i.e. the output,  when $\Gamma$  only contains fixpoint equations of the form $\pi_1{\cdot}X \doteq \pi_2{\cdot}X$ that are
non-redundant, and the rule (Output) fires. Note that the rule (FPS) represents the usual solution rule if the premise is only a single equation.

\medskip
The rules (1)--(6), and \mbox{(ElimFP)} have  highest priority;  then (MMS) and (FPS).
The rule (Output) (lowest priority) terminates an execution on $\Gamma_0$ by outputting a unifier $(\theta,\nabla',\FIX)$.
A general explanation of the vertical-bar-notation is at the end of Subsection \ref{subsec:unifalg-preparations}.

\medskip
We assume  that the algorithm \letrecunify\ halts if a failure rule (see Fig.\ref{def:failure-rules}) is applicable.
\end{definition}

\begin{figure}[!h]\small
	\begin{flushleft}
	 \mbox{(Clash)}~~ $\gentzen{\Gamma\dotcup \{s \doteq t\},\nabla,\theta~~ ~tops(s) \not= \tops(t) \mbox{ and $s$ and $t$ are not suspensions}}{\bot}$ \\[1mm]
	   \mbox{(Cycle)} ~~ $\gentzen{     
		\Gamma \cup \{\pi_1{\cdot}X_1 \doteq s_1, \ldots, \pi_n{\cdot}X_n \doteq s_n\},\nabla,\theta ~
	      ~~{\left\{ \begin{array}{@{}l} \text{ where }   s_i \text{ are not suspensions and } \\
		   X_{i+1}  \text{ occurs in }  s_i    \text{ for } \\ i = 1,\ldots,n-1 \text{ and }  X_1  \text{ occurs in } s_n.
	      \end{array} \right.}
		 }{ \bot }$ \\[3mm]
	  \mbox{(FailF)} ~~ $\gentzen{\Gamma,\nabla \cup \{a{\freshdot} a \},\theta}{\bot}$   ~~~~~  
 	  \mbox{(FailFS)}     
 	     ~~  $\gentzen{\Gamma,\nabla \cup \{a{\freshdot} X\},\theta \hspace*{6mm}  \text{ and }   a  \text{ occurs free in }  (X\theta)}{\bot}$ 
   \end{flushleft}\vspace*{-4mm}
\caption{Failure Rules of \letrecunify}\label{lrunify-rules-3}\label{def:failure-rules}
\end{figure}

Note that the two rules (MMS) and (FPS), without further precaution, may cause an exponential blow-up in the number
of fixpoint-equations (see Example \ref{example:exponential-FPS}).    
The rule (ElimFP) will bound the number of generated fixpoint equations by exploiting knowledge on operations within permutation groups.

 Note that the application of every rule can be done in polynomial time. In particular rule \mbox{(FailFS)}, since the computation of  $\FA((X) \theta)$
 can be done in polynomial time by iterating over the solution components.

\begin{example}\label{example:exponential-if-rule-missing}
We illustrate the letrec-rule (6) by a ground example  without flattening. Here we use pairs, which are functional expressions,
	i.e., $(s_1,s_2)$ means $f(s_1, s_2)$ for a binary function symbol $f$ in the language.
	
	 Let the equation be:
	$(\tletr~a.(a,b), b.(a,b)~\tin ~  b) \doteq (\tletr~b.(b,c), c.(b,c)~\tin ~  c).$
	The algorithm has to follow two possibilities for $\rho$: the identity, and the swapping $(1~2)$. \\
	We show the computation for the identity (position-)permutation $\rho$, which results in:
    $\pi = \{b \mapsto a;$ $c \mapsto b$; $a \mapsto c \}$, where the third binding that has to be added, such that the result is a bijection is
     $a \mapsto c$, which is not relevant for the result, but unique in this case. \\
	 Decomposition of the equations $(a,b) \doteq \pi{\cdot}(b,c),(a,b) \doteq \pi{\cdot}(b,c), b \doteq \pi{\cdot}c\}$
	 is possible without fail and yields only trivial equations.\\
	 The freshness constraints are $a{\freshdot}(\tletr~b.(b,c), c.(b,c)~\tin ~  c)$ and $b{\freshdot}(\tletr~b.(b,c), c.(b,c)~\tin ~  c)$, which holds.
\end{example}
\begin{example}\label{example:exponential-FPS}
	This example shows that  FPS (together with the standard and decomposition rules) may give rise to an exponential number of equations in the
	size of the original problem.
	Let there be variables $X_i, i = 1,\ldots,n$ and the equations
	$\Gamma =
	\{X_n \doteq \pi{\cdot}X_n,$ $X_n \doteq (f~X_{n-1}~\rho_n{\cdot}X_{n-1}),\ldots,  X_2 \doteq  (f~X_{1}~\rho_{2}{\cdot}X_{1})
	\}
	$
	where $\pi,\rho_1,\ldots,\rho_n$ are permutations.
	We prove that this unification problem may
	give rise to $2^{n-1}$ equations, if the redundancy rule (ElimFP) is not there. \\  
	\begin{tabular}{ll}
		The first step is by (FPS): &\hspace*{0.5mm}
		\begin{minipage}{0.6\textwidth}
			$$
			\left\{
			\begin{array}{rcl}f~X_{n-1}~\rho_n{\cdot}X_{n-1} &\doteq& \pi{\cdot}(f~X_{n-1}~\rho_n{\cdot}X_{n-1}),  \\
			X_{n-1} &\doteq& (f~X_{n-2}~\rho_{n-1}{\cdot}X_{n-2}),  \ldots
			\end{array}
			\right\}
			$$
		\end{minipage}
	\end{tabular}\\
	\begin{tabular}{ll}
		Using decomposition and inversion:&\hspace*{2mm}
		\begin{minipage}{0.5\textwidth}
			$$
			\left\{
			\begin{array}{rcl} X_{n-1}  &\doteq& \pi{\cdot} X_{n-1}, \\
			X_{n-1} &\doteq&\rho_n^{-1}{\cdot}\pi{\cdot} \rho_n{\cdot}X_{n-1}, \\
			X_{n-1} &\doteq &( f~X_{n-2}~\rho_{n-1}{\cdot}X_{n-2}), \ldots
			\end{array}
			\right\}
			$$
		\end{minipage}
	\end{tabular}\\
	\begin{tabular}{ll}
		After (FPS): & ~~~~~~~~
		\begin{minipage}{0.68\textwidth}
			$$
			\left\{
			\begin{array}{@{}rcl@{}} (f~X_{n-2}~\rho_{n-1}{\cdot}X_{n-2})&\doteq& \pi{\cdot} (f~X_{n-2}~\rho_{n-1}{\cdot}X_{n-2}), \\
			(f~X_{n-2}~\rho_{n-1}{\cdot}X_{n-2})& \doteq&\rho_n^{-1}{\cdot}\pi{\cdot} \rho_n{\cdot}(f~X_{n-2}~\rho_{n-1}{\cdot}X_{n-2}), \\
			X_{n-2} &\doteq &(f~X_{n-3}~\rho_{n-2}{\cdot}X_{n-3}), \ldots\\
			\end{array}
			\right\}
			$$
		\end{minipage}
	\end{tabular}\\
	\begin{tabular}{ll}
		Decomposition and inversion:&~~~
		\begin{minipage}{0.58\textwidth}
			$$
			\left\{
			\begin{array}{rcl}  X_{n-2}   &\doteq& \pi{\cdot}  X_{n-2}, \\
			X_{n-2}  &\doteq& \rho_{n-1}^{-1}{\cdot}\pi{\cdot}  \rho_{n-1}{\cdot}X_{n-2}, \\
			X_{n-2}  &\doteq &\rho_n^{-1}{\cdot}\pi{\cdot} \rho_n{\cdot} X_{n-2}, \\
			X_{n-2} & \doteq &\rho_{n-1}^{-1}{\cdot}\rho_n^{-1}{\cdot}\pi{\cdot} \rho_n{\cdot} \rho_{n-1}{\cdot}X_{n-2}, \\
			X_{n-2} &\doteq& (f~X_{n-3}~\rho_{n-2}{\cdot}X_{n-3}), \ldots\\
			\end{array}
			\right\}
			$$
		\end{minipage}
	\end{tabular}\medskip
 	~\\
    Now it is easy to see that all equations $X_1 \doteq \pi'{\cdot}X_1$ are generated,  with $\pi' \in \{\rho^{-1}\pi\rho$ where $\rho$ is a composition
	of a subsequence of $\rho_n,\rho_{n-1},\ldots, \rho_2\}$, which makes $2^{n-1}$ equations.
	The permutations are pairwise different using an appropriate choice of $\rho_i$ and $\pi$.
	The starting equations can be constructed  using the decomposition rule of abstractions.

\smallskip\noindent
	Without (ElimFP) all elements of the generated group of permutations have to be implicitly stored in $\Gamma$.  The rule (ElimFP) would permit to only keep a set
	of generators of the group of permutations. The explicit algorithmic treatment of generators and group operations is standard and not explained in this paper.
\end{example}

\section{Soundness, completeness, and complexity of \letrecunify}\label{sec:letrecunify:complete}

	\subsection{NP-Hardness of nominal letrec unification and matching}
	First we show that a restricted problem class of nominal letrec unification is already NP-hard.
	If the equations for unification are of the form $s_1 \doteq t_1, \ldots,s_n \doteq t_n$, and  the expressions $t_i$ do not contain variables $X_i$,
	then this is a nominal letrec-matching problem (see also Section \ref{sec:matching}).

\begin{theorem}\label{thm:matching-NP-hard}
	Nominal letrec matching (hence also unification) in $\LRL$ is $\NP$-hard,
	for two letrec expressions, where subexpressions are free of letrec.
\end{theorem}
\begin{proof}
	We encode the $\NP$-hard problem of finding a Hamiltonian cycle in a 3-regular graph
	\cite{picouleau:94,garey-johnson-tarjan:76}, which are graphs where  all nodes have the same degree $k = 3$.
	Let $G$ be a graph, $a_1,\ldots, a_n$ be the vertexes of the graph $G$, and $E$ be the set of edges of $G$.
	The first environment part is $\env_1 = a_1.(node~a_1);\ldots;a_n.(node~a_n)$, and a second environment part
	$\env_2$ consists of bindings $b. (f~a ~a')$ and $b'. (f~a' ~a)$  for every edge $(a,a') \in E$ for fresh names $b,b'$.
	Then let $t := $
	$(\tletr~\env_1;\env_2~\tin~0)$ which is intended to represent the graph. \quad
	Let the second expression encode the question whether there is a Hamiltonian cycle in a regular graph as follows:
	The first part of the environment is $\env_1' = a_1.(node~X_1),$ $\ldots,$ $a_n.(node~X_n)$. The second part is
	$\env_2'$ consisting of $b_1.(f~X_1~X_2)$; $b_2.(f~X_2~X_3)$; $\ldots;$ $b_n.(f~X_n~X_1)$, where all $b_i$ are different atoms.
	This part encodes the Hamiltonian cycle. We also need a third part that matches the edges that are not part of the Hamiltonian cycle.
	The third part $\env_3'$ consists
	of  entries of the form  $b.(f~Z~Z')$, where $b$ is always a fresh atom  for every binding, and
	$Z, Z'$  are fresh variables for every entry. Thus every such entry matches one edge. The number of these dummy entries can be computed as
	$3*n - n$ due to the assumption that the
	degree of $G$ is $3$, and the edges in the cycle are already covered by the second part.
	Let $s := $
	$(\tletr~\env_1';\env_2';\env_3'~\tin~0)$, representing the question of the existence of the Hamiltonian cycle.
	Then the  matching problem  $s\matcheq t$ is solvable iff the graph has a Hamiltonian cycle.
   	The degree  is $3$, hence it is not possible that there are shortcuts in the cycle.
\end{proof}

\subsection{Properties of the nominal unification algorithm {\letrecunify}}

We will use  $\mathit{size}(\Gamma)$ for estimating the runtime of {\letrecunify}, which is the sum of the sizes of the equated expressions, and where the size of an
expression is its size as a term tree, where we assume that names have size 1. We do not count the size of permutations.

For a non-deterministic algorithm outputting unifiers, we explain the deterministic variant that follows all proper choices, but not the don't-care choices,
and prints all solutions one after the other. This is called the {\em collecting} (and deterministic) version of the algorithm. 

\begin{theorem}\label{thm:unification-terminates}
	The decision variant of the algorithm \letrecunify\
	runs in nondeterministic polynomial time, where a single unifier
     is represented in polynomial space.   The number of rule applications is  $O(S^3\log(S))$ where $S$ is the size of the input.
	 The collecting version of \letrecunify\ returns exponentially many unifiers, where their number is bounded by $O(\exp(S^4)\log(S)))$.
\end{theorem}
\begin{proof}
   Let $(\Gamma_0,\nabla_0)$ be the  input, where $\Gamma_0$ is assumed to be flattened, and $S_{all}$ be $\mathit{size}(\Gamma_0,\nabla_0)$.
    We use $S = \mathit{size}(\Gamma_0)$ to argue on the number of steps of \letrecunify.  
	The execution of a single rule can be done in polynomial time depending on the size of the intermediate state, thus we  have to show that
	the size of the intermediate states remains polynomial and that the number of rule applications is at most polynomial.
	
	The number of fixpoint-equations for every variable $X$ is at most $S*\log(S)$
	since the  number of atoms is never increased, and
	since we assume that \mbox{(ElimFP)}  is applied whenever possible.
	The size of the permutation group on the set of all atoms in the input is at most
	$S!$, and so the length of proper subset-chains and hence the maximal number of (necessary) generators of a subgroup is at most
	$\log(S!) \leq  S*\log(S)$.  
	The redundancy of generators  can be tested in polynomial time depending on the number of atoms.
	Note also that applicability of (ElimFP) can be tested in polynomial time by checking the maximal possible subsets.
	
	\medskip
	The lexicographically ordered termination measure $(\numberVar, \numberSize,  \numberEquations)$ is used: 
\begin{enumerate}
\itemsep=0.9pt
  \item  $\numberVar$ is the number of different variables in $\Gamma$,
  \item  $\numberSize$ is the number of letrec-, $\lambda$, function-symbols and atoms
	              in $\Gamma$, but not in permutations,
 \item	 $\numberEquations$ is the number of equations in $\Gamma$.  
  \end{enumerate}

	Since shifting permutations down and simplification of freshness constraints both terminate in polynomial time, and do not increase the measures, we
	only compare states which are normal forms for shifting down permutations and simplifying freshness constraints.
	The following table shows the effect of the rules:

    The entries $+m$ represents an increase of at most $m$ in the relevant measure component. This form of table permits an easy check that
    the complexity of a single run is polynomial. Note that we omit the failure rules in the table, since these stop immediately.

{\small{
 $$
 \begin{array}{l||l|l|l|l}
              & \numberVar&\numberSize  &\numberEquations \\ \hline  
   (2)            & <    & \leq         &    <     \\
  \mbox{(FPS)}    & <    &  +2S\log(S)  &   <       \\
  \mbox{(MMS)}    & =    & <            &   +2S     \\
 (3),(4),(5),(6)  & =    & <            &   +S      \\
  \mbox{(ElimFP)} & =    &  =           & <         \\
   (1)            & \leq &  \leq        & <          \\
 \end{array}
 $$ } }

\eject

 \normalsize
The table shows  that every rule application strictly decreases the lexicographic  measure as a combination of the three basic measures.
	The entries can be verified by checking the rules, and using the argument that there are not more than $S\log(S)$ fixpoint equations for a single variable $X$.
	We use the table to argue on the (overall)  number of rule applications and hence the complexity:	
   The rules (2) and (FPS) strictly reduce the number of variables in $\Gamma$ and can be applied at most $S$ times.
	(FPS) increases the second measure at most by $2*S\log(S)$, since the number of symbols may be increased as often as there  are
	 fixpoint-equations and there are at most $S\log(S)$.  Since no other rule increases the measure, $\numberSize$ will never be greater than $2S^2\log(S)$.
	 The rule (MMS) strictly decreases $\numberSize$. Hence$ \numberEquations$, i.e. the number of equations is bounded by $4S^3\log(S)$.
	 Thus, the number of rule applications is $O(S^3\log(S))$.   \\
	 The complexity of applications of single rules is polynomial, in particular (FPS), see Section \ref{subsec:permutation-groups}.
	 The complexity of the constraint simplification (Lemma \ref{lemma:simplify-freshness-constraint}) is also polynomial.
	 We also have to argue on the failure rules. These detect all fail cases, and the size of the state part $\nabla$ remains polynomial. The checks within the failure rules
	 can be done in
	 polynomial time in $S_{all}$, where the argument for polynomiality of the check in (FailFS) is an algorithm that iteratively applies parts of $\theta$ and checks.
	
	 The arguments for the complexity of the size  a single solution are already given. Additional arguments are needed for an upper bound on the
	 number of solutions for the collecting version of  \letrecunify. An upper bound on the number of different possibilities that have to be explored at a
	 single rule application of (6) is $O(\exp(S))$.
	 In a single run of
	 \letrecunify\, the number of executions of (6) is at most  $O(S^3\log(S))$, hence we obtain an exponential bound $\exp(S^4)\log(S))$ for the number of solutions that are
	 outputted by the collecting version.
  \end{proof}

\begin{theorem}\label{thm:unification-sound-and-complete}
	The algorithm \letrecunify\  is sound and complete. I.e., every computed unifier is a unifier of the input problem (soundness),   and for every ground unifier of the
	input problem, there is a run of the (non-deterministic) algorithm that produces a unifier that has the ground unifier as  an instance.
\end{theorem}
\begin{proof}
	Soundness of the algorithm holds, by easy arguments for every rule, similar as in  \cite{urban-pitts-gabbay:03}, and since
	the letrec-rule follows the definition of $\sim$ in Def. \ref{def:alpha-equivalence-sim}. 
	A further argument is that the failure rules are sufficient to detect
	states without solutions.
	
	Completeness requires more arguments. The decomposition and standard rules, with the exception of rule (6),
	retain the set of solutions. The same for (MMS), (FPS),  and \mbox{(ElimFP)}.
	Note  that the nondeterminism in  (ElimFP) does not affect the current set of solutions.
	The nondeterministic rule (6) provides all possibilities for potential ground solutions.
	Moreover, the failure rules are not applicable to states that are solvable.
	
	A final output of  \letrecunify\ for a solvable input has at least one ground solution as instance:   
	we can instantiate all variables that remain in
	$\Gamma_{\mathit{out}}$ by a fresh atom. Then all fixpoint equations are satisfied, since the permutations cannot change this atom,
	and since the (atomic) freshness constraints hold. This ground solution can be represented in polynomial space by using $\theta$,
	plus an instance $X \mapsto a$ for all remaining  variables $X$ and a fresh atom $a$, and removing all fixpoint equations and
	freshness constraints. 
\end{proof}

\begin{theorem}\label{thm:letrec-unification-in-NP}
	The nominal letrec-unification problem is $\NP$-complete.
\end{theorem}
\begin{proof}
	This follows from Theorems \ref{thm:unification-terminates}  and \ref{thm:unification-sound-and-complete},
	and  Theorem \ref{thm:matching-NP-hard}.
\end{proof}

\section{Nominal matching with letrec: \letrecmatch}\label{sec:matching}

Reductions using reductions rules of the form $l \to r$ in higher order calculi with letrec, for example in a core-language of Haskell, 
require a (nominal) matching algorithm, matching the rules'  left hand side  to an expression or subexpression that is to be transformed.
An example is the beta-reduction (see the example below), but also a lot of other transformation rules can serve as examples.
For the application it is sufficient if the instance of the  right hand side $r\sigma$ is ground if $l\sigma$ is ground, and
the variable convention holds for $r\sigma$. 
In \cite{Fernandez-Gabbay:07} nominal rewriting (without letrec) is discussed, where more examples can be found, and where nominal matching is derived from
the nominal unification algorithm.  In this work also rewriting using freshness contexts is investigated,  and matching is defined using terms-in-context .
We only concentrate on the nominal matching part, since adding constraints is not problematic.

\begin{example}
	Consider the (lbeta)-rule, which is the version of (beta) used in call-by-need  calculi with sharing
	\cite{ariola:95,moran-sands-carlsson:99,schmidt-schauss-schuetz-sabel:08}. Note that in this case, the binding is used, but not the property ``recursive binding``.
	\begin{alignat*}{1}
	(\mathit{lbeta}) \quad & (\lambda x.e_1)~e_2 \to ~\tletr~x. e_2~\tin~e_1.
	\end{alignat*}
	An (lbeta) step, for example, reducing  $((\lambda x.x)~(\lambda y.y))$ to $(\tletr~x = (\lambda y.y)~\tin~x)$ is performed by  
	representing the target in $LRL$ and the beta-rule in the language LRLX, where $e_1,e_2$ are represented as variables $X_1,X_2$,
	and then matching
	$(\mathit{app}~(\lambda c.X_1)~X_2) \matcheq  (\mathit{app}~ (\lambda a.a)~(\lambda b.b))$, where $\mathit{app}$
	is the explicit representation of the binary application operator. This  results in $\sigma := \{X_1 \mapsto c; X_2 \mapsto \lambda b.b\}$,
	and the reduction result is the instance (using $\sigma$) of   $(\tletr~c. X_2~\tin~X_1)$, which is  $(\tletr~c. (\lambda b.b)~\tin~c)$.
	Note that this form of reduction sequences permits $\alpha$-equivalence as intermediate steps.
\end{example}

We derive a nominal letrec matching algorithm as a specialization of \letrecunify.
We use non symmetric  equations written $s \matcheq t$, where $s$ is an $\LRLX$-expression (also with permutations), and $t$ is ground, i.e., does not contain free variables.
It is easy to see that we can also assume that $t$  does not contain permutations, since these can immediately be simplified.
We assume that the input is a set $s_1 \matcheq t_1,\ldots,s_n \matcheq t_n$ of match equations of expressions, where for all $i$: $s_i$ may contain variables, and $t_i$ is ground.
We omit freshness constraints in the input, since these could be checked after the run of the algorithm.
Note that suspensions are not necessary  in the solution, and hence also no fixpoint equations.

\begin{figure}[!ht]\small
\vspace*{-1mm}
$\gentzen{\Gamma \dotcup\{e \matcheq e\}}{\Gamma}$ \qquad
$\gentzen{\Gamma\dotcup\{(f~s_1 \ldots s_n)\matcheq (f~s_1' \ldots s_n')\}}
{\Gamma \cup \{s_1 \matcheq s_1',\ldots,s_n \matcheq s_n'\}}$ \qquad
$\gentzen{\Gamma\dotcup\{\lambda a.s \matcheq \lambda a.t\}}
{\Gamma \cup \{s\matcheq t\}}$
\\[2mm]

$\gentzent{\Gamma\dotcup \{\lambda a.s\matcheq \lambda b.t\} \hspace*{6mm} a \# t}
{\Gamma \cup \{s \matcheq (a~b){\cdot}t\}}$
\qquad
$\gentzen{\Gamma\dotcup \{\pi{\cdot}X \matcheq e\}}
{\Gamma \cup  \{X \matcheq \pi^{-1}{\cdot}e\}}$
\qquad
$\gentzen{\Gamma\dotcup \{X \matcheq e_1, X \matcheq e_2\} \hspace*{6mm}  e_1 \sim e_2}
{\Gamma \cup  \{X \matcheq e_1\}}$
\\[2mm]

$\gentzen{\Gamma \dotcup \left\{\begin{array}{@{}ll@{}}  & \tletr ~a_1.s_1;\ldots;a_n.s_n~\tin~r \\ \matcheq & \tletr ~b_1.t_1;\ldots;b_n.t_n~\tin~r'
   \end{array}\right\}
   \hspace*{2.5mm}  a_i \freshdot (\tletr ~b_1.t_1;\ldots;b_n.t_n~\tin~r') \text{ for } i = 1,\ldots,n
   }
{\mathop{|}\limits_{\{\rho\}}~\Gamma \cup\{s_1 \matcheq \pi{\cdot}t_{\rho(1)},\ldots,  s_n \matcheq \pi{\cdot}t_{\rho(n)},
   r \matcheq \pi{\cdot}r'\}}$\\[2mm]
   \hspace*{6mm}\begin{minipage}{0.95\textwidth}  
       { where $\rho$ is a  permutation on $\{1,\ldots,n\}$,
    and $\pi$ is an (atom-) permutation that extends  \mbox{$\{b_{\rho(i)} \mapsto a_i \mid i = 1,\ldots,n)\}$}   with $\dom(\pi) \subseteq \{a_1,\ldots,a_n,b_1,\ldots,b_n\}$
   }
   \end{minipage}
\caption{Rules of the matching algorithm \letrecmatch}\label{fig:letrec-match}
\end{figure}

\begin{definition}
The rules of the nondeterministic algorithm \letrecmatch\ w.r.t. the language LRLX  are in Fig. \ref{fig:letrec-match} and the corresponding failure rules
are in Fig. \ref{fig:letrec-match-failure}.
The result is either a fail, or a substitution in the form $X_1 \matcheq s_1,\ldots,X_n \matcheq s_n$.
\end{definition}

\noindent
Note that the rules are designed such that permutations are moved to the rhs, and conditions can be immediately evaluated.

\begin{figure}[!h]\small
\vspace{1mm}
$\gentzen{\Gamma \dotcup \left\{\begin{array}{@{}ll@{}}  & \tletr ~a_1.s_1;\ldots;a_n.s_n~\tin~r \\ \matcheq & \tletr ~b_1.t_1;\ldots;b_n.t_n~\tin~r'
   \end{array}\right\}
   \hspace*{2.5mm} \begin{array}{l} a_i \text{ is fresh in } (\tletr ~b_1.t_1;\ldots;b_n.t_n~\tin~r')\\
    \text{ for some } i \in \{1,\ldots,n\}
   \end{array}
   }
{ \bot }$
 \\[2mm]
$\gentzen{s \matcheq t \in \Gamma, \text{ and } s \text{ is not a suspension, but }
	 \tops(s) \not= \tops(t)} {\bot}$
	  \hspace*{5mm}
 $\gentzent{\Gamma\dotcup \{\lambda a.s\matcheq \lambda b.t\} \hspace*{6mm}  a \text{ is fresh in } t)} {\bot}$
	 \\[2mm]
  $\gentzen{\Gamma\dotcup \{X \matcheq e_1, X \matcheq e_2\} \hspace*{6mm}  e_1 \not\sim e_2}
{\bot}$\vspace*{-1mm}
\caption{Failure rules of the matching algorithm \letrecmatch}\label{fig:letrec-match-failure}
\end{figure}\vspace*{-4mm}

The $\alpha$-equivalence test  $e_1 \sim e_2$ may for example  be performed as a subroutine call to this (nondeterministic) matching procedure in the collecting version,
i.e., the test succeeds if there is a nondeterministic execution with success as result.    
\begin{example} We illustrate the failure rule that signals fail, if $\lambda a.s\matcheq \lambda b.t$ is a match-equations and $a$ is fresh in $t$.
 In the case
$\lambda c.c \matcheq \lambda b.a$, the failure rule signals ``fail''. We also could proceed and replace it using the lambda-rule by $c \matcheq (c~ b){\cdot}a$,
which immediately reduces to $c \matcheq a$, which is fail sind the top symbols are different.
\end{example}

\noindent

 Standard arguments show:
\begin{theorem}\label{thm:matching-in-NP}
	\letrecmatch\  is sound and complete for nominal letrec matching.
	It decides nominal letrec matching in
	nondeterministic polynomial time.
	Its collecting  version  returns a finite complete set of an at most exponential number of matching substitutions,
	which are of at most polynomial size.
\end{theorem}

\begin{theorem}
	Nominal letrec matching is NP-complete.
\end{theorem}
\begin{proof}
	The problem is in NP, which follows from Theorem \ref{thm:matching-in-NP}. It is also NP-hard,
	which follows from Theorem \ref{thm:matching-NP-hard}.
\end{proof}

\subsection{Remarks on letrec-matching with DAGs}\label{subsec:nom-dag-matching}

A  more general situation for nominal letrec matching occurs, when the matching equations $\Gamma_0$ are compressed using a DAG.
We construct a practically more efficient algorithm \letrecdagmatch\ from \letrecunify\ as follows. First we generate $\Gamma_1$ from $\Gamma_0$,
which only contains  flattened  expressions by   
encoding the DAG-nodes as variables $X$ together with an equation $X \doteq r$ for the subexpression, which results in a unification equation.
This representation is no longer a letrec-matching problem, since there may be variables in the left and right-hand side of equations.
However, it has structural properties inherited from the sharing.
An expression is said $\Gamma_0$-ground, if it does not
reference variables from $\Gamma_0$ (also via equations).
In order to avoid suspension on the rhs (i.e. to have nicer results), the decomposition rule for $\lambda$-expressions with different binder names
is modified as follows :

\begin{flushleft}
	$\gentzen{\Gamma\dotcup(\lambda a.s \doteq \lambda b.t\},\nabla }
	{\Gamma \cup \{s \doteq (a~b){\cdot}t\}, \nabla \cup  \{a {\freshdot} t\}}$
	~~ $\lambda b.t$  is $\Gamma_0$-ground
	\\[2mm]
\end{flushleft}

The extra conditions $a {\freshdot} t$ and $\Gamma_0$-ground  can be tested in polynomial time.
The equations $\Gamma_1$ are  processed applying  \letrecunify\ (with the mentioned modification)
with the  guidance  that the right-hand sides of match-equations are also right-hand sides of equations in the decomposition rules.
The resulting matching substitutions can be interpreted as the instantiations into the variables of $\Gamma_0$.
Since $\Gamma_0$ is a matching problem, the result will be free of  fixpoint equations, and there will be no freshness constraints in the
solution.

\medskip
This construction would permits better performance than simply treating the DAG-matching problem as a unification problem.

\section{Graph-isomorphism-hardness of nominal letrec matching and unification without garbage}\label{sec:letrec-unif-match-hardness}

We will show in this section that Nominal Letrec Matching is at least as hard as Graph-Isomorphism.
  Graph-Isomorphism is known to have complexity between $\mathit{PTIME}$ and $\NP$. There are arguments that it is strictly weaker than the class of NP-complete problems
\cite{schoening:88}. There is also a claim that it is quasi-polynomial \cite{Babai:2016}, which means that it requires less than exponential time.
The general conjecture is that Graph-Isomorphism is properly between $\mathit{PTIME}$ and $\NP$.

First we clarify the notion of garbage, which is a notion from a programming language point of view.

\begin{definition}\label{def:garbage}
We say that an expression $t$ {\em contains garbage}, iff there is a subexpression $(\tletr~\env$ $ \tin ~r)$,
and the environment $\env$ can be split into two environments
$\env = \env_g;\env_{ng}$, such that $\env_g$ (the garbage) is not empty, and the atoms from
$\LA(\env_{g})$ occur neither free in $\env_{ng}$ nor in $r$.
Otherwise, the expression $t$ is {\em free of garbage} (or {\em garbage-free}).
\end{definition}

An example illustrating the notions is $\tletr~a.b;b.c~\tin~b$, where $a,b,c$ are (different) atoms. The binding $a.b$ is not used in the in-expression $b$,
hence it can be classified as garbage, whereas the binding  $b.c$ is used. In a programming language, the garbage collector would remove this binding and replace the expression
by $\tletr~b.c~\tin~b$.

Since  $\alpha$-equivalence of $\LRL$-expressions is Graph-Isomorphism-complete \cite{schmidtschauss-sabel-rau-RTA:2013}, but
$\alpha$-equivalence of garbage-free $\LRL$-expressions is polynomial, it is useful to look for improvements of unification and matching
for garbage-free expressions.

We will show that even very restricted nominal letrec matching problems are Graph-Isomorphism complete, which makes it very unlikely that there is a
polynomial algorithm.
\begin{theorem}\label{thm:matching-GI-hard}
	Nominal letrec matching with one occurrence of a single variable and a garbage-free target expression is Graph-Isomorphism-hard.
\end{theorem}
\begin{proof}
	Let $G_1, G_2$ be two regular graphs with degree $\ge 1$. Let $t$ be $(\tletr~\env_{G_1}~\tin~ g~ b_1~b_1' \ldots b_m~b_m')$ the encoding of an arbitrary graph
	$G_1$ where $\env_{G_1}$ is the encoding
	as in the proof of Theorem \ref{thm:matching-NP-hard}: nodes are
	encoded as $a_1,\ldots,a_n$,i.e., the bindings in the environment are $a_1.(\mathit{node}~a_1),\ldots,a_n1.(\mathit{node}~a_n)$,
	and the edges are encoded as  follows: The $i$-th edge  $(a_k~a_l)$ is encoded as $b_i.(\mathit{edge}~a_k~a_l),b_i'.(\mathit{edge}~a_l~a_k)$.
	Then $t$ is free of garbage, since the graph $G_1$ is regular.
	Let the environment $\env_{G_2}$  be the encoding of $G_2$ in $s = (\tletr~ \env_{G_2}~\tin~X)$, where $\env_{G_2}$ is constructed in the same way from the graph $G_2$.
	Then $s$ matches $t$ iff the graphs $G_1,G_2$ are isomorphic.
	If there is an isomorphism of $G_1$ and $G_2$, then it is easy to see that this bijection leads to an equivalence of the environments,
	and we can instantiate $X$ with $(g ~b_1~b_1' \ldots b_m~b_m')$.
	Since the graph-isomorphism problem for regular graphs of degree $\ge 1$ is GI-hard \cite{Booth:78},
	we have shown GI-hardness.
\end{proof}

\section{On fixpoints and garbage}\label{sec:no-garbage-fixpoints}

We will show in this section that $\LRLX$-expressions without garbage only have trivial fix-pointing permutations.
Looking at Example \ref{example:fixpoints-possible}, the  $\alpha$-equivalence  $(a~b)\cdot (\tletr~c.a;d.b~\tin ~\mathit{True})$ $\sim$
$(\tletr~\allowbreak c.a;d.b ~\tin ~\mathit{True})$ holds, where $\dom((a ~b)) \cap \FA(\tletr~c.a;d.b~\tin ~\mathit{True}) = \{a,b\} \not= \emptyset$.
 However, we see that the expressions is equivalent (in a programming language semantics) to $\mathit{True}$, and that the whole environment in this example is garbage,
 since no binding is
 referenced from the in-expression
  (see Def.  \ref{def:garbage}).

\medskip
As a helpful information, we write the $\alpha$-equivalence-rule for letrec-expressions in the ground language $\LRL$ as an extension
of the rule for lambda-abstractions.\smallskip

\begin{flushleft}
  $\gentzen  {r \sim \pi{\cdot}r',\, s_i \sim \pi{\cdot}t_{\rho(i)},\, i = 1,\ldots,n , \hspace*{4mm}
        M {\freshdot}(\tletr ~b_1.t_1;\ldots;b_n.t_n~\tin~r') }
   {\tletr ~a_1.s_1;\ldots;a_n.s_n~\tin~r \sim \tletr ~b_1.t_1;\ldots;b_n.t_n~\tin~r'}$ \\[1mm]
  \vspace{1mm}\hspace*{5mm} \begin{minipage}{0.9\textwidth}
       where $\rho$ is a  permutation on $\{1,\ldots,n\}$,  $M = \{a_1,\ldots,a_n\} \setminus \{b_1,\ldots,b_n\}$, and $\pi$
     is an atom-permutation-extension of    
      the bijective function $\{b_i \mapsto a_{\rho(i)},i = 1,\ldots,n \}$ such that
      $\dom(\pi) \subseteq (\{b_1,\ldots,b_n\} \cup \{a_1,\ldots,a_n\})$.
   \end{minipage}
     \\[3mm]
 \end{flushleft}

 Note that $\alpha$-equivalence of $s,t$  means structural equivalence of $s,t$ as trees, and a justification always comes with a bijective relation between
the positions of $s,t$ where only the names of atoms at nodes may be different.

A further example of garbage is $(\tletr~a.0; b.1~\tin~(f~b))$, where $a$ is unused, but $b$ is used in the right hand side. In this case  $a.0$
is garbage.
 Another example is $e := (\tletr~a.d; b.1; c.d~\tin~(f~b))$, which is an example with a free atom $d$, and the garbage consists of two bindings, $\{a.d;c.d\}$.
 It is $\alpha$-equivalent to $(\tletr~a'.d; b.1; c'.d~\tin~(f~b)) =: e'$.  Note that in this case, there are two
  different permutations (bijective functions) mapping $e'$ to (the $\alpha$-equivalent) $e$:
 $\{a' \mapsto a;c' \mapsto c\}$ and $\{a' \mapsto c;c' \mapsto a\}$.

   The next lemma shows that  this situation is only possible if the expressions contain garbage.

\begin{lemma}\label{lemma:gcfree-unique}
If $s\sim t$, and $s$ is free of garbage, then $\alpha$-equivalence provides a unique correspondence of the positions of $s$ and $t$.
\end{lemma}
\begin{proof}
The proof is by induction on the structure and size of expressions.
For the structure, the only nontrivial case is  letrec:
Let  $s  = (\tletr~a_1. e_1,\ldots,a_n.e_n~\tin~e)$ $\sim$  $(\tletr~b_1. f_1,\ldots,b_n.f_n~\tin~f) = t$. Note that due to syntactic
equality all permutations of the environments are also to be considered.
Then there is bijective mapping $\varphi$,
with $\varphi(b_i) = a_{\rho(i)}, i=1,\ldots,n$, where $\rho$ is a permutation on $\{1,\ldots,n\}$, and such that
$e_i \sim \varphi(f_{\rho(i)}), i=1,\ldots,n$, $e \sim \varphi(f)$, and $(\{a_1,\ldots,a_n\}\setminus\{b_1,\ldots,b_n\})\# t$ holds.
Let $\overline{\varphi}$ be the atom-permutation
that extends $\varphi$, mapping $(\{a_1,\ldots,a_n\}\setminus\{b_1,\ldots,b_n\})$ to $(\{b_1,\ldots,b_n\} \setminus \{a_1,\ldots,a_n\})$.

The induction hypothesis implies a unique position correspondence
of $e$ and $f$, since $e \sim \overline{\varphi}(f)$. This implies that the bindings for  $\{a_1,\ldots,a_n\} \cap \FA(e)$ have a unique
correspondence to the bindings in $t$. This is continued by exhaustively following free occurrences of atoms
$a_i$ in the right hand sides of the top bindings in $s$. Since there is no garbage in $s$, all bindings can be reached by this process, hence
we have uniqueness of the correspondence of positions.
\end{proof}

\begin{proposition}\label{prop:no-gc-implies-fp-condition}
Let $e$ be an expression that does not have garbage, and let $\pi$ be a permutation. Then $\pi{\cdot}e \sim e$ implies
   $\dom(\pi) \cap \FA(e) = \emptyset$.
\end{proposition}
\begin{proof}
The proof is by induction on the size of the expression.
\begin{itemize}
\itemsep=0.9pt
  \item  If $e$ is an atom, then this is trivial.
  \item If $e = f~e_1 \ldots.e_n$, then no $e_i$ contains garbage, and $\pi{\cdot}e_i \sim e_i$   implies $\dom(\pi) \cap \FA(e_i) = \emptyset$,
  hence also $\dom(\pi) \cap \FA(e) = \emptyset$.
  \item If $e = \lambda a.e'$, then there are two cases:

     \begin{enumerate}
     \itemsep=0.9pt
       \item $\pi(a) = a$. Then $\pi{\cdot}e' \sim e'$, and we can apply the induction hypothesis.
       \item $\pi(a) = b \not= a$. Then $(a~b){\cdot}\pi$ fixes $e'$, and $b \# e'$.
       The induction hypothesis implies $\dom((a~b){\cdot}\pi) \cap \FA(e') = \emptyset$.
       We have $ \dom(\pi) \subseteq \dom((a~b){\cdot}\pi))  \cup \{a,b\}$, hence $\dom(\pi) \cap \FA(\lambda a.e') = \emptyset$.
     \end{enumerate}
   \item First a simple case with one binding in the environment: $t = (\tletr~a_1.e_1~\tin~e)$,
     $\pi{\cdot} t \sim t$.   
     If $\pi(a_1) =  a_1$, then $\pi{\cdot}(e,e_1) \sim (e,e_1)$,
      and the induction hypothesis implies $\dom(\pi) \cap \FA(e,e_1) = \emptyset$, which in turn implies $\dom(\pi) \cap \FA(t) = \emptyset$.

     If $\pi(a_1) = b \not=  a_1$, then $b \# (e,e_1)$  and for $\pi' := (a_1~b){\cdot}\pi$, it holds $\pi'{\cdot}(e,e_1) \sim (e,e_1)$, and so
      $\dom((a_1~b){\cdot}\pi) \cap \FA(e,e_1) = \emptyset$.  since $\dom(\pi) \subseteq \dom((a_1~b){\cdot}\pi) \cup \{a_1,b\}$,
      we obtain $\dom(\pi) \cap t  = \emptyset$.
      In the case of one binding, it is irrelevant whether the binding is garbage or not.

   \item Let $t = (\tletr~a_1.e_1;\ldots;a_n.e_n~\tin~e)$, and $t$ is a fixpoint of $\pi$, i.e. $\pi(t) \sim t$.  Note that no part of the  environment is garbage.
      The permutation $\pi$ can be split into $\pi =\pi_1{\cdot}\pi_2$, where $\dom(\pi_1) \subseteq \FA(t)$ and $\dom(\pi_2) \cap \FA(t) = \emptyset$.
      From $t \sim \pi{\cdot}t$ and Lemma \ref{lemma:gcfree-unique} we obtain  that there is a unique permutation $\rho$ on $\{1,\ldots,n\}$, such that
      there is an injective mapping $\varphi: \pi(a_1) \mapsto a_{\rho(1)}, \ldots, \pi(a_n) \mapsto a_{\rho(n)}$,  and $e \sim \varphi\pi(e)$,
      $e_{\rho(i)} \sim \varphi\pi(e_i)$. Then $\alpha$-equivalence implies that $\varphi\pi$ can be extended to a atom-permutation $\overline{\varphi}\pi$ by mapping the
      atoms in $\{a_1, \ldots, a_n\} \setminus \{\pi(a_1),\ldots,\pi(a_n)\}$ bijectively to  $\{\pi(a_1),\ldots,\pi(a_n)\} \setminus \{a_1, \ldots, a_n\}$.
      By the freshness constraints for $\alpha$-equivalences of letrec-expressions,   $\overline{\varphi}\pi(e) = \varphi\pi(e)$ and $\overline{\varphi}\pi(e_i) = \varphi\pi(e_i)$
      which in turn implies that $e \sim  \overline{\varphi}\pi(e)$ and $e_i \sim  \overline{\varphi}\pi(_ie)$, and we can apply the induction hypothesis.

       This shows that  $\FA(e) \setminus \{a_1,\ldots,a_n\}$  are not moved by  $\overline{\varphi}\pi$, and the same for all $e_i$, hence this also holds
    for $t$.   
\end{itemize}

\vspace*{-6mm}
\end{proof}
\begin{corollary}\label{cor:no-garbage-no-fixpointeqs}
Let $e$ be an expression that does not have garbage, and let $\pi$ be a permutation.
Then $\pi{\cdot}e \sim e$ is equivalent to
      $\dom(\pi) \cap \FA(e) = \emptyset$.
\end{corollary}
 \begin{proof}
 This follows from Proposition \ref{prop:no-gc-implies-fp-condition}. The other direction is easy.
 \end{proof}

The proof also shows a slightly more general statement:
\begin{corollary}
Let $e$ be an expression such that in all environments with at least two bindings there are no garbage bindings, and let $\pi$ be a permutation.
Then $\pi{\cdot}e \sim e$ is equivalent to
      $\dom(\pi) \cap \FA(e) = \emptyset$.
\end{corollary}

 In case that  the input does not represent garbage-parts, and the semantics is defined such that only ground garbage free expressions are permitted,
 the set of rules in the case without atom-variables can be optimized as follows: (ElimFP) can be omitted and instead of (FPS) there are two rules:

 \begin{flushleft}

  \mbox{(FPS2)}~$\gentzen{\Gamma\dotcup \{X \doteq \pi{\cdot}X\}, \nabla}
   {\Gamma,   \nabla \cup \{a{\#}X \mid a \in \dom(\pi)\}}$, \quad
   \\[2mm]
 \mbox{(ElimX)}~$\gentzen{\Gamma\dotcup \{X \doteq e\}, \theta}
   {\Gamma,  \theta \cup \{X \mapsto e\}}$, \quad
    \begin{minipage}{0.6 \textwidth}  if    
     $X \not\in \Var(\Gamma)$,
   and $e$ is not a suspension of $X$.
   \end{minipage}
   \\[2mm]

 \end{flushleft}

 \begin{example} It cannot be expected that the letrec-decomposition rule (7) can be turned into a deterministic rule, and
 to obtain a  unitary nominal unification, under the restriction that  input expressions are garbage-free, and also instantiations are garbage-free.
 Consider the equation:
 $$(\tletr ~a_1.e_1; a_2.e_2~\tin~ ((a_1,a_2),X)) \doteq (\tletr~b_1.f_1; b_2.f_2~ \tin~(X’,(b_1,b_2))).$$
 Then the in-expressions do not enforce a unique correspondence between the bindings of the left  and right-hand bindings.
 An example also follows from the proof of Theorem \ref{thm:matching-GI-hard}, which shows that  even nominal matching may have several incomparable
 solutions for garbage-free expressions.
 \end{example}

\section{Nominal unification with letrec and atom-variables}\label{section:LRA}

  In this section we extend the unification algorithm to the language $\LRLXA$, which is an extension of  $\LRLX$ with atom variables.
  Atom-variables increase the expressive power of a term language  with atoms alone. If in an application example it is known that in a pair
  $(x_1,x_2)$ the expressions $x_1,x_2$  can only be atoms,
  but $x_1 = x_2$ as well as $x_1 \not= x_2$ is possible, then two different unification problems have to be formulated. If atom variables are possible,
  then the notation
  $(A_1,A_2)$ covers both possibilities.

  It is known that the nominal unification problem with atom-variables but without letrec is NP-complete \cite{schmidt-schauss-sabel-kutz:19}.
  An algorithm and corresponding rules and discussions can be found  in \cite{schmidt-schauss-sabel-kutz:19}.
  An implication is NP-hardness of nominal unification with atom variables and letrec.

 \subsection{Extension with atom-variables}

As an extension of  $\LRLX$, we define the language$\LRLXA$ as follows:
Let $A$ denote atom variables, $V$ denote atom variables or atoms, $W$ denote suspensions of atoms or atom variables, $X$ denotes expression variables,
$\pi$ a permutation, and $e$ an expression.
The syntax of the language $\LRLXA$ is
\[\begin{array}{lcl}
V  & ::=&   a \mid A   \\
W  & ::=&   \pi \cdot V\\
\pi & ::=&  \emptyset \mid  (W~W)  \mid \pi{\circ}\pi\\
 e & ::=&  \pi{\cdot}X  \mid W  \mid  \lambda W.e \mid (f~e_1~\ldots e_{\ari(f)})~| ~(\tletr~W_1.e_1; \ldots;W_n.e_n~\tin~e)
    \end{array}
\]
Let $\Var(e)$ be the set of atom or expression variables occurring in $e$, and let $\AtVar(e)$ be the set of atom  variables occurring in $e$.
 Similarly for sequences of expressions or permutations.\\   
The expression $\pi{\cdot}e$ for a non-variable expression  $e$ means an operation, which is performed by shifting $\pi$ down in the expression,
  using the simplifications $\pi_1 {\cdot} (\pi_2 {\cdot} X)$ $\to $ $(\pi_1 \circ  \pi_2) {\cdot}X$,
where only  expressions $\pi{\cdot}X$ and $\pi{\cdot}V$ remain, where the latter are called  {\em suspensions} and where $\pi{\cdot}V$ is abbreviated as $W$.
\begin{remark} An {\bf alert} for the reader: In this section the use of atom-variables induces generalizations and changes in the
$\LRLXA$-formulation of problems: binders may now be
suspensions of atom-variables, and also
  ``nested" permutation representations are permitted, which is due to atom variables, since in general, this permutation representation
 cannot be simplified. \\
 Several simple facts and intuitions that are used for $\LRLX$ no longer hold.
\end{remark}

A {\em freshness constraint} in our unification algorithm is of the form $V{\freshdot}e$  where $e$ is an $\LRLXA$-expression.
The justification for the slightly more complex form as usual ($a \freshdot X$) is that atom variables prevent a simplification to this form.
The notation $\pi^{-1}$ is
defined as the reversed list of swappings of $\pi$, where the single (perhaps complex) swappings are not modified.  A semantical justification for this
inverse is by checking the ground instantiations. 

We also view $\pi\cdot V \freshdot e$ as identical to the constraint $V \freshdot \pi^{-1} \cdot e$.

Naively applying ground substitutions may lead to syntactically invalid ground expressions, since instantiation may make binding atoms in letrecs equal,
which is illegal. This will be prevented by freshness constraints.

\begin{example}
The equation
 $$(\mathit{app}~(\tletr~A.a, B.a  ~\tin~ B)~A) \doteq (\mathit{app}~(\tletr~A.a, B.a~\tin~ B)~B)$$
 enforces that $A,B$ are instantiated with the same atom, which contradicts the syntactic assumption on distinct atoms for the
  binding names in letrec-expressions. This will be dealt with by adding the freshness constraint $A\#B$.
  However,
 $$(\mathit{app}~(\tletr~A.a, C.a  ~\tin~ C)~A) \doteq (\mathit{app}~(\tletr~A.a, D.a~\tin~ D)~B)$$
is solvable. Note that the additional freshness constraints are $A\#C,A\#D$.
%
\end{example}

\begin{remark}\label{remark:prevent-illegal-letrecs}
We circumvent the problem of illegal ground instances by adding  for every letrec-expression in the input of the unification  algorithm
sufficiently many freshness constraint
that prevent these illegal expressions (see below). It is sufficient to prevent equal binding names in every single letrec-environment.
\end{remark}

\begin{definition}
An {\em $\LRLXA$-unification problem}  is a pair $(\Gamma, \nabla)$, where  $\Gamma$ is a set of equations $s \doteq t$,
 and $\nabla$ is a set of freshness constraints $V \# e$.
 In addition, for every letrec-subexpression $\tletr~W_1.e_1$, $\ldots,$ $W_m.e_m~\tin~e$, which occurs in $\Gamma$ or $\nabla$,
 the set $\nabla$ must also contain the freshness constraint  $W_i \freshdot W_j$ for all $i,j = 1,\ldots,m$ with $i \not= j$.

A {\em (ground) solution} of $(\Gamma, \nabla)$  is a substitution $\rho$ (mapping variables in $\Var(\Gamma, \nabla)$ to ground expressions),
 such that $s\rho \sim t\rho$ for all equations $s \doteq t$ in $\Gamma$, and
 for all $V{{\freshdot}}e \in \nabla$:  $V\rho \# (e\rho)$ holds.

The {\em decision problem}  is whether there is a solution for a given $(\Gamma, \nabla)$.
\end{definition}

\begin{proposition}\label{prop:LRLXA-in-NP}
The  $\LRLXA$-unification problem is in NP, and hence NP-complete.
\end{proposition}
\begin{proof}
The argument is that every ground instantiation of an atom variable is an atom, which can be guessed and checked in polynomial time:
guess the images of atom variables under a ground solution $\rho$ in the set of atoms in the current state, or in an arbitrary
 set of fresh atom variables
of cardinality at most the number of different atom variables in the input.  Then  instantiate accordingly 
thereby removing all atom-variables.
The resulting problem can be decided (and solved) by an NP-algorithm as shown in this paper (Theorem \ref{thm:unification-terminates}).
\end{proof}

 \begin{remark}\label{remark-equalityAB} Note that the equation  $A = \pi{\cdot}B$ for atom variables $A,B$ can be encoded as the freshness constraint $A \# \lambda \pi{\cdot}B. A$.
 In the following we may use equations $V_1 =_{\#} \pi{\cdot}V_2$ as a more readable version of $V_1 \# \lambda \pi{\cdot}V_2. V_1$.
\end{remark}

 \begin{figure}[!b]\small
 \vspace{2mm}
\begin{flushleft}
 $(1)~\gentzen{\Gamma \dotcup\{e \doteq e\}}{\Gamma}$  \qquad
%
    (2)$~\gentzen{\Gamma \dotcup\{\pi_1{\cdot}V_1 \doteq \pi_2{\cdot}V_2\},\nabla, \theta}{\Gamma,  \nabla \cup \{V_1 =_{\#} \pi_1^{-1}\pi_2{\cdot}V_2\}, \theta}$ \\[2mm]
%
  $(3a)~\gentzen{\Gamma \dotcup\{\pi_1{\cdot}X \doteq \pi_2{\cdot}Y\},\nabla, \theta  \qquad X \not= Y}
     {\Gamma[\pi_1^{-1}\pi_2{\cdot}Y/X],  \nabla[\pi_1^{-1}\pi_2{\cdot}Y/X], \theta \cup \{X \mapsto \pi_1^{-1}\pi_2Y\}}$
   \quad \\[2mm]
     $(3b)~\gentzen{\Gamma \dotcup\{\pi_1{\cdot}X \doteq \pi_2{\cdot}V\},\nabla, \theta  }
     {\Gamma[\pi_1^{-1}\pi_2{\cdot}V/X],  \nabla[\pi_1^{-1}\pi_2{\cdot}V/X], \theta \cup \{X \mapsto \pi_1^{-1}\pi_2V\}}$
   \quad \\[2mm]

$(4)~ \gentzen{\Gamma\dotcup\{(f~(\pi_1{\cdot}X_1) \ldots (\pi_n{\cdot} X_n)) \doteq (f~(\pi'_1{\cdot} X'_1) \ldots (\pi'_n{\cdot} X'_n))\}}   
   {\Gamma \cup \{\pi_1 {\cdot} X_1 \doteq \pi'_1 {\cdot} X_1',\ldots,\pi_n{\cdot} X_n \doteq \pi'_n{\cdot} X'_n\}}$\\[2mm]        

$(5)~ \gentzen{\Gamma\dotcup\{(\lambda W.\pi_1 {\cdot} X_1 \doteq \lambda W.\pi_2 {\cdot} X_2\}}   
      {\Gamma \cup \{\pi_1 {\cdot} X_1\doteq \pi_2 {\cdot} X_2\}}$    
    \\[2mm]
     $(6)~ \gentzen{\Gamma\dotcup(\lambda W_1.\pi_1 {\cdot} X_1 \doteq \lambda W_2.\pi_2 {\cdot} X_2\},\nabla}
      {\Gamma \cup \{\pi_1 {\cdot} X_1\doteq (W_1~W_2){\cdot}\pi_2 {\cdot} X_2\},\nabla \cup  \{W_1 {\freshdot}(\lambda W_2.\pi_2 {\cdot} X_2)\}}$
     \\[3mm]

$(7)~ \gentzen{\Gamma \dotcup \left\{
	\begin{array}{l}
	\tletr ~W_1.\pi_1{\cdot}  X_1;\ldots;W_n.\pi_n {\cdot} X_n~\tin~\pi {\cdot} Y \doteq \\
	\tletr ~W'_1.\pi'_1 {\cdot} X'_1;\ldots;W'_n.\pi'_n {\cdot} X'_n~\tin~\pi' {\cdot} Y'
	\end{array}
	 \right\}, \nabla}
       {\begin{array}{l}
         {\begin{array}{ll}
         \hspace*{-1mm}
         {\begin{array}{c}
               \left|\begin{array}{c}
                ~\\~\\~\\
               \end{array}\right.
               \\
           {\{\rho\}}
           \end{array}}
         &
         \left( \begin{array}{l}
             \Gamma \,\mathop{\cup}\, \left\{
              \begin{array}{l}
                 ~ \mathrm{decompose}(n\!+\!1,  \lambda W_1 \ldots \lambda W_n.(\pi_1 {\cdot} X_1,\ldots,\pi_n {\cdot} X_n, \pi {\cdot} Y) \\
                    \hspace*{0.7cm}\doteq  \lambda W'_{\rho(1)}.\ldots \lambda W'_{\rho(n)}.(\pi'_{\rho(1)} {\cdot} X'_{\rho(1)},\ldots,\pi'_{\rho(n)} {\cdot} X'_{\rho(n)}, \pi' {\cdot} Y'))
              \end{array}\right\},   \\
              ~ \\[-5mm]
              \nabla \cup
              \left\{
              \begin{array}{l}
                 ~ \mathrm{decompfresh}(n\!+\!1,  \lambda W_1 \ldots \lambda W_n.(\pi_1 {\cdot} X_1,\ldots,\pi_n {\cdot} X_n, \pi {\cdot} Y) \\
                    \hspace*{0.7cm}\doteq  \lambda W'_{\rho(1)}.\ldots \lambda W'_{\rho(n)}.(\pi'_{\rho(1)} {\cdot} X'_{\rho(1)},\ldots,\pi'_{\rho(n)} {\cdot} X'_{\rho(n)}, \pi' {\cdot} Y'))
              \end{array}\right\}
              \end{array}     \right)
             \end{array}}
                  ~\\[2mm]
             \mbox{where $\rho$ is a  permutation on $\{1,\ldots,n\}$ and $\mathrm{decompose}(n,.)$ is the equation part of $n$-fold } ~\\
               \mbox{ application of rules (4), (5) or (6)  and $\mathrm{decomposefresh}(n,.)$ is the freshness constraint} ~\\
               \mbox{ part of the $n$-fold application of rules (4), (5) or (6); (in both cases after flattening).}  \\
         \end{array}}
      $
  \end{flushleft}\vspace*{-4mm}
 \caption{Standard  and decomposition rules with atom variables of \letrecunifyA.}\label{LRA-Fig1}\vspace{2mm}
 \end{figure}

 \begin{figure}[!th]\small
  \begin{flushleft}
 (MMS), (FPS), (ElimFP) and (Output) are almost the same as the ones in Fig \ref{lrunify-rules-2}. \\[2mm]
 \mbox{(MMS)} $\gentzen{\Gamma\dotcup \{\pi_1{\cdot}X \doteq e_1, \pi_2{\cdot}X \doteq e_2\},\nabla}
	{\Gamma \cup  \{\pi_1{\cdot}X \doteq e_1\} \cup \Gamma', \nabla \cup \nabla'}$,~~~
	\begin{minipage}{0.50 \textwidth} if  $e_1, e_2$ are not suspensions,
	  where $\Gamma'$ is the set of equations generated by decomposing $\pi_1^{-1}{\cdot}e_1 \doteq \pi_2^{-1}{\cdot}e_2$ using (1)--(7),
	  and where $\nabla'$ is the corresponding resulting set of freshness constraints. \\
	\end{minipage}
	\\[-1mm]
	\mbox{(FPS)}~$\gentzen{\Gamma\dotcup \{\pi_1{\cdot}X \doteq \pi'_1{\cdot}X, \ldots,\pi_n{\cdot}X \doteq \pi'_n{\cdot}X, \pi{\cdot}X \doteq e\}, \theta}
	{\Gamma \cup  \{\pi_1\pi^{-1}{\cdot}e \doteq \pi'_1\pi^{-1}{\cdot}e, \ldots,\pi_n\pi^{-1}{\cdot}e \doteq \pi'_n\pi^{-1}{\cdot}e\},  \theta \cup \{X \mapsto \pi^{-1}{\cdot}e\}}$, \\
	\hspace*{1.1cm} \begin{minipage}{0.9 \textwidth}  
		If $X \not\in \Var(\Gamma,e)$,
		and $e$ is not a suspension, and (Cycle) (see Fig.\ref{def:failure-rules}) is not applicable.
	     \end{minipage}
	\\[2mm]
		\mbox{(ElimFP)}~~$\gentzen{\Gamma\dotcup \{\pi_1{\cdot}X \doteq \pi'_1{\cdot}X, \ldots,\pi_n{\cdot}X \doteq \pi'_n{\cdot}X, \pi{\cdot}X \doteq \pi'{\cdot}X\}, \theta}
	{\Gamma\cup \{\pi_1{\cdot}X \doteq \pi'_1{\cdot}X, \ldots,\pi_n{\cdot}X \doteq \pi'_n{\cdot}X\}, \theta}$,   \\
	\hspace*{1.5cm} \text{If } $\pi^{-1}\pi' \in \gen{\pi_1^{-1}\pi_1,\ldots,\pi_n^{-1}\pi_n},$ \\
	\hspace*{1.5cm} \text{and } $\pi_i,\pi_i', \pi,\pi'$ are ground, i.e. do not contain atom variables.
	\\[2mm]
	
	\mbox{(Output)}~~$\gentzen{\Gamma, \nabla, \theta}
	{(\theta, \nabla, \Gamma)}$     
	\begin{minipage}{0.45\textwidth} if $\Gamma$ only consists of fixpoint-equations.
	\end{minipage}

      \mbox{(ElimA)}~~$\gentzen{\Gamma,\nabla, \theta}
       {\begin{array}{ll}
         {\begin{array}{c}    
              {\left|\begin{array}{c}
                     \end{array}\right.}
               \\
           {\left\{ \begin{array}{ll}   \text{atoms in } \Gamma,\nabla,\theta \\ \text{and a fresh atom } a   \end{array}\right\}}
           \end{array}}
         &
            \Gamma[a/A], \nabla[a/A], \theta \cup \{A \mapsto a\}
       \end{array}
      }$
 \end{flushleft}\vspace*{-3mm}
 \caption{Main rules   of \letrecunifyA}\label{LRA-Fig2}
 \end{figure}

 \begin{figure}[!h]\small
 \vspace*{2mm}
	\begin{flushleft}
	 \mbox{(Clash)}~~ $\gentzen{\Gamma\dotcup \{s \doteq t\},\nabla,\theta~~ ~tops(s) \not= \tops(t) \mbox{ and $s$ and $t$ are not suspensions}}{\bot}$ \\[1mm]
	  \mbox{(ClashA)}~~ $\gentzen{\{s \doteq t\} \mbox{ is in } \Gamma,\mbox{ and }  ~~
	    \begin{array}{l}  s \text{ is a suspension of an atom or atom variable}\\
	     \mbox{ and  $\tops(t)$ is a function symbol, $\lambda$ or letrec }
	     \end{array}}
	     {\bot}$
	     \\[1mm]
	      \mbox{(Clashab)}~~ $\gentzen{\Gamma\dotcup \{a \doteq b\},\nabla,\theta~\hspace*{5mm} a \not= b}{\bot}$ \\[1mm]
	   \mbox{(Cycle)} ~~ $\gentzen{     
		\begin{array}{l}\text{If } \pi_1{\cdot}X_1 \doteq s_1, \ldots, \pi_n{\cdot}X_n \doteq s_n  \text{ in }  \Gamma
		\text{ where }   s_i \text{ are not suspensions} \\
		\text{ and }  X_{i+1}  \text{ occurs in }  s_i  \text{ for } i = 1,\ldots,n-1 \text{ and }  X_1  \text{ occurs in } s_n.
		 \end{array}}{ \bot }$ \\[3mm]
	  \mbox{(FailF)} ~~ $\gentzen{ a{\freshdot} a \in \nabla}{\bot}$   ~~~~~  
 	  \mbox{(FailFS)}     
 	     ~~  $\gentzen{ a{\freshdot} X \in \nabla ~~~ \text{ and }   a  \text{ occurs free in }  (X\theta)}{\bot}$ 
   \end{flushleft}\vspace*{-6mm}
\caption{Failure Rules of \letrecunifyA}\label{lrunify-rules-A}\label{def:failure-rules-with-A}
\end{figure}

\subsection{Rules of the algorithm \letrecunifyA}

Now we describe the nominal unification algorithm {\letrecunifyA} for $\LRLXA$.
It will extend the algorithm {\letrecunify} by a treatment of atom variables that extend the expressibility.
It has flexible rules, such that a strategy can be added
to control the nondeterminism and such that it is an improvement over a brute-force guessing-algorithm that first guesses all atom instances of atom-variables
 and then uses  Algorithm \letrecunify\ (see  Algorithm \ref{alg:lrunifAB}
for such an improvement). 
The simple idea is to only make these guesses if a certain space-bound of the whole state is exceeded and then use the guesses and further rules to shrink
the size of the problem representation.  
%
\noindent Note that permutations with atom variables may lead to an exponential blow-up of their size due to iterated application of rules, which is defeated by a compression mechanism.
 Note also that
equations of the form $A \doteq e$, in particular $A \doteq \pi{\cdot}A'$, cannot be solved by substitutions ($A \mapsto \pi{\cdot}A'$) for two reasons:
(i) the atom variable $A$ may occur in the right hand side, and (ii) due to our compression mechanism (see below),
the substitution may introduce cycles into the compression, which is forbidden.

Atoms in the input are permitted. In the rules an extra mention of atoms is only in (2), (3), (ElimFP), (ElimA), (Clashab), (FailF), (FailFS)   and  in (ElimFP).

\begin{definition} The algorithm
{\letrecunifyA}  operates on a tuple $(\Gamma, \nabla, \theta)$, where the rules are defined in Figs. \ref{LRA-Fig1}  and  \ref{LRA-Fig2},
and failure rules are in Fig. \ref{def:failure-rules-with-A}.

\medskip
\noindent The rules (7)  and (ElimA) are don't know non-deterministic, whereas the other ones are don't care non-deterministic.
The following explanations are in order:
\begin{enumerate}
\itemsep=0.9pt
\item $\Gamma$ is assumed to be a set of flattened equations $e_1 \doteq e_2$ (see the remarks after Definition \ref{def:LRLX-unifier}).
\item We assume that  $\doteq$ is symmetric,
\item $\nabla$ contains freshness constraints, like $a\freshdot e$, $A\freshdot e$,  which in certain cases may be written as equations of the form
$A =_{\#} \pi{\cdot}A'$ (see Remark \ref{remark-equalityAB}, for better readability and simplicity).    
\item $\theta$ represents the already computed substitution as a list of replacements of the form $X \mapsto e$.
We assume that the substitution is the iterated replacement.
Initially $\theta$  is empty.
\end{enumerate}
The final state will be reached, i.e. the output,  when $\Gamma$  only contains fixpoint equations of the form $\pi_1{\cdot}X \doteq \pi_2{\cdot}X$, and the rule (Output) fires.

In the notation of the rules, we will use $[e/X]$  as substitution that replaces $X$ by $e$. We  may omit $\nabla$ or $\theta$ in the notation of a rule, if they are not changed.
We will also use a  notation ``$|$'' in the consequence part of rule (6), where all possibilities for $\rho$ have to be considered (denoted as the set $\{\rho\}$), to denote
 disjunctive  (i.e. don't know) nondeterminism.  There are two nondeterministic rules with disjunctive nondeterminism: the letrec-decomposition rule (7)  exploring all alternatives of the correspondence between bindings;
 the other one is (ElimA) that guesses the instantiation of an atom-variable. In case it is guessed to be different from all currently used atoms, we
 remember  this fact (for simplicity) by selecting a  fresh atom for instantiation.
 The other rules can be applied in any order, where it is not necessary to explore alternatives.
\end{definition}

 We assume that permutations in the algorithm {\letrecunifyA} are compressed using a grammar-mechanism,  as a variation of
     grammar-compression in \cite{lohrey-maneth-schmidt-schauss:12,gascon-godoy-schmidt-schauss:12}.
     However, we do not mention it in the rules of the algorithm, but we will use it in the complexity arguments (see below).

  The use of the iterated decomposition in rule (7) appears clumsy at a first look, however, it is an easy algorithmic representation of the method to define the
  permutations (with atom variables) in a recursive fashion, where the introduction of permutation variables is avoided.

 \begin{definition} 
  The components of a {\em permutation grammar} $G$, used for compression, are:
 \begin{itemize}
 \itemsep=0.85pt
   \item Nonterminals $P_i$.
   \item For every nonterminal $P_i$ there is an associated inverse $P_j$, which can also be written as $\overline{P}_i$.
   \item Rules of the form $P_i \to w_1 \ldots w_n$, $n \geq 1$ where $w_i$ is either a nonterminal or a terminal. At all times
              $\overline{P}_i \to \overline{w}_n \ldots \overline{w}_1$ holds, i.e., if a nonterminal is added its inverse is added accordingly.
              Usually, $n \le 2$, but also another fixed bound for $n$ is possible.
   \item Terminal elements are $\emptyset$,  $(V_1~V_2)$.
 \end{itemize}
 The grammar is deterministic: every nonterminal is on the left-hand side of exactly one rule. It is also non-recursive: the terminal index is such that $P_i$ can only be in right-hand sides of the nonterminal $P_j$ with $j<i$.
 The function  $inv$, mapping $P_i \to \overline{P}_i$ and $T \to T$ for terminals $T$ computes the inverse in constant time.
 This is true by construction, because if $P \to w_1 \ldots w_n$ then $inv(P) \to \inv(w_n) \ldots \inv(w_1)$ and $\inv(T) = T$ for terminals.
 Every nonterminal $P$ represents a permutation $\val(P)$, which is  computed from the grammar as follows:
 \begin{enumerate}
 \itemsep=0.85pt
 \item $\val(P) = \val(w_1)~\ldots~\val(w_n)$ (as a composition of permutations),  if $P \to w_1 \ldots w_n$.
   \item $\val(\emptyset) = \mathit{Id}$.
   \item $\val((P_1{\cdot}V_1~P_2{\cdot}V_2)) = (\val(P_1){\cdot}V_1~~\val(P_2){\cdot}V_2)$.
 \end{enumerate}
 \end{definition}

 \begin{lemma} For nonterminals $P$ of a permutation grammar $G$, the  permutation $\val(\inv(P))$ is the inverse of $\val(P)$.
  \end{lemma}

 \subsection{Arguments for correctness and completeness}

 Let $S$ denote in the following the size of the initial unification problem.
 \begin{proposition}\label{prop:grammar-simplify-poly} Let $G$ be a permutation grammar, and let $P$ be a nonterminal, such that $\val(P)$  contains $n$ atoms,
 and does not contain any atom variables.
     Then $\val(P)$ can be transformed into a permutation of length at most $n$ in polynomial time.
 \end{proposition}
 \begin{proof}
 For every $P$ the size of the set $At(P)$ has an upper bound $S$ and can be computed in time $O(S \cdot \log(S))$
 For every such atom $a \in At(P)$ we compute its image $P\cdot a$ and save the result in a mapping from atoms to atoms. The computation of $P \cdot a$ can be done in
  $O(S^2)$, yielding a total of $O(S^3)$ for the construction of this map, which has size $O(S)$. At last, the construction of the permutation list can be done
  in linear time, i.e. $O(S)$.
 \end{proof}

 Now we consider the operations to extend the grammar during the unification algorithm. By extension we mean to add non-terminals and rules to the grammar,
 where the grammar is used as a compression device.   

\begin{proposition}\label{prop:grammar-poly}
Extending  $n$ times the grammar $G$  can be performed in polynomial time in $n$, and the size of the initial grammar $G$.
\end{proposition}
 \begin{proof}
 We check the extension operations: \\
  Adding a nonterminal can be done in constant time.
    Adding an inverse of $P$ is in constant time, since the inverses of the sub-permutations are already available.
   Adding a composition $P = P_1 {\cdot} P_2$ and at the same time the inverse, can be done in constant time.
\end{proof}

This polynomial upper bound will be  used in the Proof of Theorem \ref{thm:LRLXA-strategies-complete}.

 \noindent As a summary we obtain:
Generating the permutation grammar  on the fly during the execution of the unification rules can be done in polynomial time, since (as we will show below)
the number of rule executions is polynomial in the size of the initial input.
Also the  operation
of applying a compressed ground permutation to an atom is polynomial.

      Note that (MMS) and (FPS), without further precaution, may cause an exponential blow-up in the number
     of fixpoint equations (see Example \ref{example:exponential-FPS}).    
  The rule (ElimFP) will limit the number of fix\-point equations for atom-only permutations by exploiting knowledge on operations on permutation groups.
  The rule (ElimA) can be used according to a dynamic strategy (see below): if the space re\-quire\-ment for the state is too high, then it can be applied until
  simplification rules make $(\Gamma,\nabla)$  smaller.

  The rule (Output) terminates an execution on $\Gamma_0$ by outputting a unifier $(\theta,\nabla',{\cal X})$,
  where the solvability of $\nabla'$ needs to be checked using methods as in the algorithm proposed in \cite{schmidt-schauss-sabel-kutz:19}.
  The method is to nondeterministically instantiate atom-variables by atoms, and then checking the freshness constraints, which is  in NP
  (see also Theorem \ref{thm:unification-terminates}).

\medskip
We will show that the algorithm runs in polynomial time  by applying (ElimA) following a strategy defined below.
  There are two rules, which can lead to a size increase of the unification problem if we ignore the size
  of the permutations:   (MMS) and (FPS):     
    \begin{itemize}
    \itemsep=0.95pt
    \item{(MMS)} Given the equations $X \doteq e_1, X \doteq e_2$, the increase of the size of $\Gamma$ after the application of the rule has an upper bound  $O(S)$.

    \item{(FPS)} Given $X\doteq \pi_1 {\cdot} X,  \ldots, X\doteq \pi_k {\cdot} X, X\doteq e$, the size increase has an upper bound $O(S)$. Disregarding the
     permutations of only atoms,
     it is not known whether there exists a polynomial upper bound of the number of independent permutations with atom variables - but it seems very unlikely.
   \end{itemize}

\begin{definition}
Let $p(x)$ be some easily computable function $\bbbr^+ \to \bbbr^+$.   
The rule $\ELIMAB(p)$ is defined as follows:
 \begin{quote}
  $\ELIMAB(p)$: ~~~ If there are $k>p(S)$ fixpoint equations $X \doteq \pi_1  {\cdot}X, \ldots,  X \doteq \pi_k  {\cdot}X$ in $\Gamma$ for some variable $X$,
      then apply (ElimA) for all $A \in \AtVar(\pi_1,\dots,\pi_k)$. Then immediately apply (ElimFP) exhaustively.     
 \end{quote}
\end{definition}

\begin{definition}\label{alg:lrunifAB}
The guided version  \letrecunifyAB$(p)$ of  {\letrecunifyA} is obtained by replacing  (ElimA) with $\ELIMAB(p)$ where $p(x)$ is some (easily computable) function $\bbbr^+ \to \bbbr^+$,
such that  $\forall x\in \bbbr^+: q(x) \ge p(x) \ge x*\log(x)$ holds for some polynomial $q$.
In addition the priority of the rules is as follows, where highest priority comes first:   (1), \ldots, (6), \mbox{(ElimFP)}, \mbox{(MMS)}, \mbox{(Output)}.
Then   $\ELIMAB(p)$, \mbox{(FPS)}, and the nondeterministic rule (7) with lowest priority.
\end{definition}

\begin{lemma} \label{lem:fpslimit}  Let $\Gamma, \nabla$ be a solvable input. For every function $p(x)$ with  $\forall x \in \bbbr^{+}: p(x)\geq x\log(x)$,
  the algorithm \letrecunifyAB$(p)$ does not get stuck,
 and for every intermediate state of the algorithm \letrecunifyAB$(p)$ it holds that
  the number of fixpoint equations per expression variable is
  bounded above  by $p(S)$. 
\end{lemma}
\begin{proof}
The upper bound of the number of fixpoint equations is proved as follows:
Let $m$ be the number of atoms in the original unification problem. The rule (ElimA) (called by  (ElimAB)) introduces at most $S-m$ new atoms, which implies
at most $S$ atoms at any time.
If \letrecunifyAB$(p)$ exceeds its upper space bound and applies $\ELIMAB(p)$
on the fixpoint equations $X \doteq \pi_1 {\cdot} X, \ldots, X \doteq \pi_k {\cdot} X$, the number of fixpoint equations of $X$ can be reduced to at most $S\log(S) \leq p(S)$
(see the proof of Theorem \ref{thm:unification-terminates}).

\medskip
Since the input is solvable, the choices can be made accordingly, guided by the solution, and then it is not possible that
there is an occurs-check-failure for the variables. Hence if the upper line of the  preconditions of (FPS) is a part of $\Gamma$, there will also be a maximal variable $X$,
such that the condition $X \not\in \Var(\Gamma,e)$ can be satisfied. 
\end{proof}

The following theorem shows that the (non-deterministic) algorithm for nominal unification with letrec and atom-variables can be guided by a strategy that instantiates atom-variables
only if the number of fixpoint equations grows too large. The problem is that with atom-variables we could not exhibit a redundancy eliminating rule for fixpoint-constraints as for the case with atoms.
The algorithm \letrecunifyAB$(p)$ provides this compromise. It guesses the instantiation of certain atom-variables if the number of fixpoint equations is greater than a bound.
This strategy prevents for example an exponential growth of the number of fixpoint-equations. There is flexibility through the choice of a threshold-function. Thus Theorem
\ref{thm:LRLXA-strategies-complete} shows that
with a threshold function satisfying only weak conditions, the algorithm can be controlled  and that there is a chance to find a good practical compromise between too much non-determinism
and space-explosion.

The algorithm is sound and also complete, however, we do not provide  explicit arguments here.

\sloppy
\begin{theorem}\label{thm:LRLXA-strategies-complete} Let $\Gamma, \nabla$ be a solvable input.  For every function $p(x)$ such that there is a polynomial $q(x)$ with
 $\forall x: q(x) \geq p(x) \geq x \log(x)$, \letrecunifyAB$(p)$   
does not get stuck and runs in polynomial space and time.  
\end{theorem}
\begin{proof}
The proof is inspired by the  proof of Theorem \ref{thm:unification-terminates}, and uses Lemma \ref{lem:fpslimit} that shows that the number of fixpoint-equations for a single variable is at most $p(S)$.

\medskip
Below we show some estimates on the size and the number of steps. 
	The termination measure $(\numberVar, \numberSize,  \numberEquations, \numberEqsNonX)$,
  	 which is ordered lexicographically,  is as follows:
\begin{description}
\itemsep=0.9pt
\item[$\numberVar$] is the number of different variables in $\Gamma$,
 \item[$\numberSize$] is the number of letrec-, $\lambda$, function-symbols and atoms
	              in $\Gamma$, but not in permutations,
  \item[$\numberEquations$] is the number of equations in $\Gamma$, and
  \item[$\numberEqsNonX$] is the number of equations where non of the equated expressions is a variable.
  \end{description}
		
	Since shifting permutations down and simplification of freshness constraints both terminate and do not increase the measures, we
	only compare states which are normal forms for shifting down permutations and simplifying freshness constraints.
	
	The following table shows the effect of the rules:
	Let $S$ be the size of the initial $(\Gamma_0, \nabla_0)$ where $\Gamma $ is already flattened.
Again,  the entries $+W$ represent a size increase of at most $W$ in the relevant measure component.

{\small{ $$
 \begin{array}{l||l|l|l|l}
              & \numberVar&\numberSize  &\numberEquations & \numberEqsNonX  \\ \hline   
   (3)            & <    & \leq         &    =        &   \leq\\
  \mbox{(FPS)}    & <    &  +2p(S)      &    <        &  + 2p(S)  \\
  \mbox{(MMS)}    & =    & <            &   +2S       &     =       \\
 (4),(5),(6),(7)   & =    & <            &   +S        &    \leq     \\
 \mbox{\ELIMAB(p)} & =   &  =           & <           &  \leq    \\
   (1)            & \leq &  \leq        & <           &   \leq     \\
   (2)            & =    &  =           & =           & <   \\
 \end{array}
 $$ } }
 	The table shows  that the rule applications strictly decrease the measure.
	The entries can be verified by checking the rules, and using the argument that there are not more than $p(S)$ fixpoint equations for a single
	variable $X$.
	We use the table to argue on the number of rule applications and hence the complexity:
	 The rules (3) and (FPS) strictly reduce the number of variables in $\Gamma$ and can be applied at most $S$ times.
	 The rule (FPS) increases the second measure at most by $2p(S)$, since the number of symbols may be increased as often as there are
	 fixpoint-equations, and there are at most $p(S)$.  Thus the measure $\numberSize$ will never be greater than $2Sp(S)$.
	
\medskip
	 The rule (MMS) strictly decreases $\numberSize$, hence $\numberEquations$, i.e. the number of equations, is bounded by $4S^2p(S)$.
	 The same bound holds for $\numberEqsNonX$.
	 Hence the number of rule applications is $O(S^2p(S))$.
       Of course, there may be a polynomial effort in executing a single rule, and by Proposition \ref{prop:grammar-poly}
        the contribution of the grammar-operations
        is also only polynomial.
      Finally, since $p(x)$ is polynomially bounded by $q(x)$, the algorithm can be executed in polynomial time.
\end{proof}

\section{Nominal letrec matching with environment variables}\label{sec:atom-letrec-match}

 We extend the language $\LRLXA$ by variables $\Env$ that may encode (partial) letrec-environments for a nominal matching algorithm,
 which leads to a larger coverage of
practically occurring nominal matching problems    
in reasoning about the (small-step operational) semantics of programming languages with letrec.  

\begin{example}
	Consider as an example a rule  (llet-e) of the operational semantics of a functional core language, which merges \tletr-environments
	(see \cite{schmidt-schauss-schuetz-sabel:08}):
	$(\tletr~\Env_1~\tin~ (\tletr~\Env_2~\tin~ X))~ \to~ (\tletr~\Env_1;\Env_2~\tin~ X).$
	It can be applied to an expression $(\tletr~a.0;b.1~\tin~(\tletr~c.(a,b,c)~\tin~c))$ as follows:
	The left-hand side $(\tletr~\Env_1~\tin$~ $(\tletr~\Env_2~\tin~ X))$ of the reduction rule   matches
	$(\tletr~a.0;b.1~\tin$ $(\tletr$ $c.(a,b,c)~\tin~c))$ with the match:
	$\{\Env_1 \mapsto \{a.0;b.1\}; \Env_2 \mapsto \{c.(a,b,c)\}; X \mapsto c\}$,
	producing the next expression as an instance of the right hand side $(\tletr~\Env_1;\Env_2~\tin~ X)$, which is
	$(\tletr~a.0;b.1;c.(a,b,c) ~\tin~c)$.
	Note that for application to extended lambda calculi,
	more care is needed \wrt scoping in order to get valid reduction results in all cases. The restriction
	that a single letrec environment binds different variables becomes more important.
	The reduction (llet-e) is  correctly applicable, if the target expression satisfies the so-called distinct variable convention,
	i.e., if all bound variables are different and if all free variables in the expression are different from all bound variables.
	In this section we will add freshness constraints that enforce different binders in environments.
	
	An alternative that is used for a  similar unification task in \cite{schmidt-schauss-sabel:16}
	requires the additional construct of non-capture constraints: $\mathit{NCC}(\env_1, \env_2)$,
	which means that for every valid instantiation $\rho$,
	variables occurring free in $env_1\rho$ are not captured by the top letrec-binders in  $env_2\rho$.
	In this paper we focus on nominal matching for the extension with environment variables, and leave the investigation of
	reduction rules and sequences for further work.
\end{example}

\begin{definition}
	The grammar for the extended language $\LRLXAE$ ({\bf L}et{\bf R}ec {\bf L}anguage e{\bf X}tended with {\bf A}tom variables and
	{\bf E}nvironment) variables $E$ is:

		\[\begin{array}{lcl}
  V  & ::=&   a \mid A   \\
  W  & ::=&   \pi \cdot V\\
   \pi & ::=&  \emptyset \mid  (W~W)  \mid \pi{\circ}\pi\\
    e & ::=&  \pi{\cdot}X  \mid W  \mid  \lambda W.e \mid (f~e_1~\ldots e_{\ari(f)})~| ~(\tletr~\env~\tin~e) \\
    \env & ::= & E \mid W.e \mid \env;\env  \mid \emptyset       
    \end{array}
\]
\end{definition}

\begin{figure}[!b]\small
\begin{flushleft}
 $(1)~\gentzen{\Gamma \dotcup\{e \matcheq  e\}}{\Gamma}$  \qquad
%
  $(2)~\gentzen{\Gamma \dotcup\{\pi_1{\cdot}A \matcheq a\},\nabla, \theta  }
     {\Gamma[\pi_1^{-1}{\cdot}a/A],  \nabla[\pi_1^{-1}{\cdot}a/A] , \theta \cup \{A \mapsto \pi_1^{-1}{\cdot}a\}}$     
   \quad \\[2mm]
     $(3)~\gentzen{\Gamma \dotcup\{\pi_1{\cdot}X \matcheq  e\},\nabla, \theta  }
     {\Gamma[\pi_1^{-1}{\cdot}e/X],  \nabla[\pi_1^{-1}{\cdot}e/X], \theta \cup \{X \mapsto \pi_1^{-1}{\cdot}e\}}$
   \qquad 
%
%
$(4)~ \gentzen{\Gamma\dotcup\{(f~e_1 \ldots e_n)) \matcheq (f~e_1' \ldots e_n')\}}   
   {\Gamma \cup \{e_1 \matcheq e_1',\ldots e_n \matcheq e_n'\}}$  
\\[2mm]
$(5)~ \gentzen{\Gamma\dotcup\{(\lambda a.e_1 \matcheq \lambda a.e_2\}}   
      {\Gamma \cup \{e_1 \matcheq e_2\}}$    
   \qquad %
    $(6)~ \gentzen{\Gamma\dotcup\{(\lambda W.e_1 \matcheq \lambda a.e_2\},\nabla}   
      {\Gamma \cup \{(W~a){\cdot}e_1 \matcheq e_2\}, \nabla \cup \{a \freshdot \lambda W.e_1\}}$    
     \\[3mm]
 $(7)~ \gentzen{\Gamma \dotcup \left\{
	\begin{array}{l}
	\tletr ~W_1.e_1;\ldots;W_n.e_n~\tin~ e \matcheq \\
	\tletr ~a_1.e'_1;\ldots;a_n.e'_n~\tin~e'
	\end{array}
	 \right\}, \nabla  \qquad
	    \begin{array}{l} \mbox{If the left hand side environment} \\
	                    \mbox{does not contain environment variables.}
	    \end{array}
	 }
       {\begin{array}{l}
         {\begin{array}{ll}
         \hspace*{-1mm}
         {\begin{array}{c}
               \left|\begin{array}{c}
                ~\\~\\~\\
               \end{array}\right.
               \\
          \hspace*{-3mm} {\{\rho\}}
           \end{array}}
         &
         \left( \begin{array}{l}
             \Gamma \,\mathop{\cup}\, \left\{
              \begin{array}{l}
                 ~ \mathrm{decompose}(n\!+\!1,  \lambda W_1 \ldots \lambda W_n.(e_1,\ldots,e_n, e)) \\
                    \hspace*{0.7cm}\matcheq  \lambda a_{\rho(1)}.\ldots \lambda a_{\rho(n)}.(e'_{\rho(1)},\ldots,e'_{\rho(n)}, e'))
              \end{array}\right\},   \\
              ~ \\[-5mm]
              \nabla \cup
              \left\{
              \begin{array}{l}
                 ~ \mathrm{decompfresh}(n\!+\!1,  \lambda W_1 \ldots \lambda W_n.(e_1,\ldots,e_n, e)) \\
                    \hspace*{0.7cm}\doteq  \lambda a_{\rho(1)}.\ldots \lambda a_{\rho(n)}.e'_{\rho(1)},\ldots,e'_{\rho(n)}, e'))
              \end{array}\right\}
              \end{array}     \right)
             \end{array}}
                  ~\\[2mm]
             \mbox{where $\rho$ is a  permutation on $\{1,\ldots,n\}$ and $\mathrm{decompose}(n,.)$ is the equation part of $n$-fold } ~\\
               \mbox{ application of rules (4), (5) or (6)  and $\mathrm{decomposefresh}(n,.)$ is the freshness constraint} ~\\
               \mbox{ part of the $n$-fold application of rules (4), (5) or (6).}  \\
         \end{array}}
      $
      \\[3mm]
  $(8)~ \gentzen{\Gamma\dotcup\{\tletr ~W_1.e_1;\ldots;E; \ldots;W_n.e_n~\tin~e \matcheq
	\tletr ~a_1.e'_1;\ldots;a_n.e'_m~\tin~e'\},\nabla,\theta}   
   {\begin{array}{@{}l@{}}
    \begin{array}{@{}l@{}l@{}}
         \hspace*{-1mm}
         {\begin{array}{@{}c@{}}
               \left|\begin{array}{@{}c@{}}
                ~\\
               \end{array}\right.
               \\
           {\hspace*{-3mm}\{\sigma\}}
           \end{array}}
         &
     ((\Gamma{\cup}\{\tletr ~W_1.e_1;\ldots;E;\ldots;W_n.e_n~\tin~e \matcheq
	 \tletr~a_1.e'_1;\ldots;a_n.e'_m~\tin~e'\})\sigma, \nabla \sigma, \theta{\cup}\sigma
	 ~\\[2mm]
           &  \mbox{where  $\sigma = \{E \mapsto A_1.X_1, \ldots  A_k.X_k\}$  where $A_i, X_i$ are fresh variables and $k \leq m-n$.}
      \end{array}
      \end{array}}    $
 \end{flushleft}\vspace*{-3mm}
\caption{Standard  and decomposition matching rules with environment variables of \letrecenvmatch.}\label{LRAE-match1}
 \end{figure}

We define a nominal matching algorithm, where in addition environment variables may
occur (also non-linear) in left hand sides, but not in the right hand sides.

\medskip
The matching algorithm with environment variables is described below. It can be obtained from the algorithm {\letrecunifyA} by adding a rule
that (nondeterministically) instantiates environment variables by environments of the form $W_1.X_1; \ldots;W_k.X_k$. This
   can eliminate all environment variables. After this operation of eliminating all environment variables,
   it is possible to use the algorithm {\letrecunifyA}. However, since the equations are match-equations, it is possible to
   derive simplified and optimized rules of {\letrecunifyA}.  We describe the rules explicitly, in order to exhibit the optimization
   possibilities of a matching algorithm compared with a unification algorithm.

\begin{definition}\label{def:matchalg-env}
	The matching algorithm  \letrecenvmatch\   is described in Fig. \ref{LRAE-match1}.
	Permitted inputs are
	matching equations between expressions, i.e. variables are only permitted in left hand sides of (matching) equations.
	The don't know-nondeterminism
	is indicated in the respective rules.\\
	It is assumed that in the input as well as after instantiating the $\env$-variables, the freshness constraints contain constraints
	that prevent that a letrec-environment contains bindings with the same binder (see Remark \ref{remark:prevent-illegal-letrecs}).
	The result is a substitution, a freshness constraint and a substitution.
\end{definition}

We omit failure rules, since these obviously follow from the nominal matching algorithm.
 Guessing the number of instances into environment variables may lead to clashes due to a wrong number of bindings in environments.
 An implementation can be more clever by checking the possible number of bindings before guessing.

 It is easy to see that the problem itself is in NP, by the following argument:
 Guess atom-variables in the left hand side, where we only have to choose from the already existing atom-variables in the problem and a fresh atom,
 and iterate this until all atom-variables are replaced. Then we guess the
   environment variables in a general way, as in rule (8), where the  number of bindings is at most the maximal number of bindings in the
   letrec-environments in the right hand side. This (non-deterministic) guessing and replacement is polynomial.
   Then we can apply Theorem \ref{thm:letrec-unification-in-NP}.
 Since the rules of  {\letrecenvmatch} are simplified rules of the algorithm {\letrecunifyA}, we obtain:

\begin{theorem}\label{thm:letrecenvmatch}
  The nominal matching algorithm  {\letrecenvmatch} is sound and complete and runs in NP time.
\end{theorem}

\section{Conclusion and future research}\label{sec:conclusion}

We construct  nominal unification algorithms for expressions with letrec, for the case where only atoms are permitted,
and also for the case
where in addition atom variables are permitted. We also describe several nominal letrec matching algorithms for variants, in particular also for
expressions with environment variables.
All algorithms run in (nondeterministic) polynomial time.
Future research is to investigate  extensions of nominal unification with environment variables $\Env$,
perhaps as an extension of the matching algorithm.

Future work is also an investigation into the connection with equivariant nominal unification
\cite{cheney-JAR:2010,cheney-diss:2004,aoto-kikuchi:16},
and  to investigate nominal matching together with equational theories.
Also applications of nominal techniques to reduction steps in operational semantics of calculi with letrec
and  transformations should be more deeply investigated. Also
 nominal unification in the combination of  letrec, environment variables and atom variables is subject to future research.

\subsection*{Acknowledgements}

The research of  Manfred Schmidt-Schau{\ss} was partially supported by the Deutsche Forschungsgemeinschaft (DFG)
under grant SCHM 986/11-1.

\smallskip
\noindent The research for the author Temur Kutsia was partially supported by the Austrian Science  Fund (FWF) project P 28789-N32.

\smallskip
\noindent The research of Jordi Levy was partially supported by the MINECO/FEDER
projects RASO (TIN2015-71799-C2-1-P) and LoCoS (TIN2015-66293-R).

\smallskip
\noindent The research of Mateu Villaret was partially supported by UdG project MPCUdG2016/055.

\smallskip
\noindent We thank the reviewers for their detailed comments that greatly helped to improve the paper.

\end{document}